\documentclass[a4paper,11pt]{amsart}
\usepackage{wiaspreprint}

\usepackage[pdftex]{hyperref}
\hypersetup{%
pdftitle = {The Propagation-Separation Approach: Consequences of Model Misspecification}
pdfauthor = {Saskia Becker}
 pdfkeywords = {Structural adaptive smoothing, Propagation, Separation, Local likelihood, Exponential families, Memory, model misspecification},
 pdfstartview = {FitH},
 colorlinks = {false},
 breaklinks = {true},
 pdfborder = {0 0 0},
 linkcolor = {blue},
  plainpages = {false},
}

\usepackage[square]{natbib}
\usepackage{url}
\usepackage{amsmath, amssymb, amsfonts}
\usepackage{color}
\usepackage{caption}
\usepackage{subcaption}
\usepackage{paralist}
\newcommand{\pkg}[1]{{\normalfont\fontseries{b}\selectfont #1}}

\allowdisplaybreaks[4] 
\usepackage{amsthm}
\numberwithin{equation}{section}
\theoremstyle{definition}
\newtheorem{Def}{Definition}[section]
\newtheorem{Alg}{Algorithm}
\newtheorem{Not}[Def]{Notation}
\newtheorem{Ass}{Assumption}
\theoremstyle{plain}
\newtheorem{Thm}[Def]{Theorem}
\newtheorem{Prop}[Def]{Proposition}
\newtheorem{Lem}[Def]{Lemma}
\newtheorem{Claim}[Def]{Claim}
\theoremstyle{remark}

\newtheorem{Cor}[Def]{Corollary}
\newtheorem{Rem}[Def]{Remark}
\newtheorem{Ex}[Def]{Example}

\begin{document}

\author{%
\nofnmark{}				% keine Fu"snotenmarke
Saskia Becker
\footnote{Weierstrass Institute \\
Mohrenstr. 39\\ 10117 Berlin \\ Germany\\
E-Mail: saskia.becker@wias-berlin.de}
%\\
%E-Mail: peter.mathe@wias-berlin.de}
%Peter Math\'{e} %\footnote{Institution 2}
%\addmark{,}{1}				% Trennzeichen, Marke als Bezug auf schon
%					% vorhandene Marke
}
\title{The Propagation-Separation Approach: \\ Consequences of model misspecification}

% in Englisch: Kleinschreibung, nach Doppelpunkt oder Bindestrich mit Grossbuchstaben beginnen
\nopreprint{1877}               	% Preprint-Nummer
%%\nopreyear{+++Preprint-Jahr+++}	% Jahr des Preprints
\selectlanguage{english}		% hier nicht veraendern, wichtig fuer Datumsformat
\date{\today}			% Datum fixieren - Schreibweise z.B. October 7, 2009
	% Math. Subject Classif.
%\pacs[2008]{}				% ggf. Physics Astronomy Classif.
%\thanks{}
%\thanks{}				% Danksagungen; Punkt am Satzende erscheint automatisch!
%% amsart only! abstract befor maketitle
%\begin{abstract} \end{abstract}
%% article and other classes: abstract after maketitle

%
\begin{abstract}
The article presents new results on the Propagation-Separation Approach by \citet{PoSp05}. 
This iterative procedure provides a unified approach for nonparametric estimation, supposing a local parametric model.
The adaptivity of the estimator ensures sensitivity to structural changes.
Originally, an additional memory step was included into the algorithm, where most of the theoretical properties were based on.
However, in practice, a simplified version of the algorithm is used, where the memory step is omitted.
% and the algorithm still shows the desired behavior.
Hence, we aim to justify this simplified procedure by means of a theoretical study and numerical simulations.
In our previous study \citep{BecMat2013}, we analyzed the simplified Propagation-Separation Approach, 
supposing piecewise constant parameter functions with sharp discontinuities. 
Here, we consider the case of a misspecified model.
\end{abstract}

\maketitle

\thispagestyle{empty}

\section{Introduction}

In statistics, local modeling is one of the most commonly used approaches for nonparametric estimation, see, for instance, \citet{MR1391963} and \citet{MR1383587}.
Local models can be described by weights which depend on the explanatory variables (design) only.
An alternative approach for local modeling is based on weighting schemes that depend (additionally) on the response variables (observations).
This helps to avoid blurring at discontinuities.
As it turned out, the comparison of noisy observations in single points suffers from a lack of robustness, see \citet{buades2005} and the references therein.
Therefore, the Propagation-Separation Approach by \citet{PoSp05} uses a multiscale approach with iteratively updated weights that benefit from the previously aggregated information about the underlying structure.
This enables the detection of discontinuities. 
Within homogeneous regions the method yields similar results as non-adaptive smoothing.

The Propagation-Separation Approach relates to Lepski's method \citep{MR1091202, MR2240642}.
Furthermore, it extends the Adaptive Weights Smoothing (AWS) procedure \citep{PoSp00}, whose theoretical properties were restricted to additive Gaussian noise. 
In contrast, the Propa\-gation-Sepa\-ration Approach supposes a local likelihood model.
Hence, it is applicable to a large variety of problems.
It has been successfully applied in the context of 
image denoising \citep{poas, MR2853730, Li2012, PoSp08, Tabelow08}, 
time series analysis \citep{PoDi08}, density estimation, and classification \citep{PoSp05}, for example.
Despite the practical use of this method, only few properties are known.
The aim of this article is to provide a better understanding of the Propagation-Separation Approach, the involved parameters, its theoretical properties, and its behavior in practice.

For the verification of theoretical properties, we suppose a local exponential family model. 
This provides an explicit expression of the Kullback-Leibler divergence on which the algorithm is based.
The model includes, for instance, the Gaussian regression and the inhomogeneous Ber\-noulli, exponential, and Poisson models \citep[\S 2]{PoSp05}.
In practice, the procedure only requires a metric on the design space and the existence of an appropriate approximation of the Kullback-Leibler divergence. 

For the sake of computational simplicity the algorithm is formulated with respect to a local constant model. 
However, it can be generalized to local linear and local polynomial models as well. 
In our previous study \citep{BecMat2013}, we concentrated on the case of piecewise constant functions with sharp discontinuities. 
Here, we will analyze consequences of a misspecified model.
As in \citep{BecMat2013}, we omit the additional memory step and avoid Assumption (S0) in \citep{PoSp05}.
Assumption (S0) is problematic since it requires the data-driven weights of the estimator to be statistically independent of the observations. 
The memory step was included into the algorithm in order to ensure a certain stability of estimates.
However, in applications of the Propagation-Separation Approach it has been omitted.
As it turned out, its practical use is questionable, while the algorithm provides the desired behavior even without the memory step.

The outline is as follows. 
First, we will recall the local exponential family model and the original algorithm of the Propagation-Separation Approach.
Then, we will introduce a parameter choice strategy for the adaptation bandwidth that allows the verification of propagation and a certain stability of estimates for functions with bounded variability within well-separated regions.
Moreover, we will define an \emph{associated step function} which approximates the estimation function for sufficiently large location bandwidths.
In Section~\ref{sec:inhomPCPractice}, we will provide further details concerning the practical application of the newly introduced inhomogeneous propagation condition.
Our subsequent numerical simulations illustrate the formation of the associated step function.
All examples suggest the convergence of the Propagation-Separation Approach.
However, this property could not be proven theoretically for reasons that we will discuss in Section~\ref{sec:PS-discussion}. 
In Appendix~\ref{app:auxiliary}, we will recall some auxiliary results by \citet{PoSp05} and \citet{BecMat2013}.
Longer proofs will be given in Appendix~\ref{app:proofs}.

\section{Model and methodology}\label{sec:Model}

We assume a local parametric model, more precisely the local likelihood model.
This general setting enables a unified approach to a broad class of nonparametric estimation problems.

\begin{Not}[Setting]\label{not:setting}
Let $\mathcal{P} := \lbrace \mathbb{P}_{\theta} \rbrace_{\theta \in \Theta}$ denote a parametric family of probability distributions with a convex parameter set~$\Theta \subseteq \mathbb{R}$, where $(\Omega, \mathcal{F}, \mathbb{P}_{\theta})$ forms, for every $\theta \in \Theta$, a probability space with dominating $\sigma$-finite measure~$\mathbb{P}$. 
We consider a metric space~$\mathcal{X}$ with metric~$\delta$, and a measurable observation space~$(\mathcal{Y}, \mathcal{B})$, where $\mathcal{Y} \subseteq \mathbb{R}$ and~$\mathcal{B}$ denotes the Borel algebra.
On the deterministic design $\lbrace X_i \rbrace_{i=1}^n \subseteq \mathcal{X}$ with $n \in \mathbb{N}$, we observe the statistically independent random variables~$\lbrace Y_i \rbrace_{i=1}^n$, where $Y_i \sim \mathbb{P}_{\theta(X_i)} \in \mathcal{P}$ and $Y_i(\omega) \in \mathcal{Y}$, $\omega \in \Omega$, for every $i \in \lbrace 1,...,n \rbrace$.
Then, we aim to estimate, the unknown parameter function $\theta: \mathcal{X} \to \Theta \subseteq \mathbb{R}$ on the design~$\lbrace X_i \rbrace_{i=1}^n$, that is~$\lbrace \theta_i \rbrace_{i=1}^n$ with $\theta_i := \theta(X_i)$.
\end{Not}

For the sake of simplicity, we assume the design to be known and the observation space as well as the parameter set to be one-dimensional, that is $\mathcal{Y}, \Theta \subseteq \mathbb{R}$.
Basically, the Propagation-Separation Approach can be applied on any measurable vector space~$\mathcal{Y} \subseteq M$ with $Y_i \sim \mathbb{P}_{\theta(X_i)}$ for every $i \in \lbrace 1,...,n \rbrace$ and $\theta: \mathcal{X} \to \Theta \subseteq M$, where~$M$ is endowed with a possibly asymmetric distance function. 

The algorithm is based on the Kullback-Leibler divergence. 
This has an explicit expression under the following assumption, which was supposed in \citep{BecMat2013, PoSp05}, as well. 
A list of parametric families that satisfy this assumption is given in \citep[Table 1]{BecMat2013}. 
We use the common notation
\[
	C^2( \Theta, \mathbb{R} ) := 
	\left\{ f: \Theta \to \mathbb{R}: \text{ the first and second derivative of } f \text{ exist and are continuous} \right\}.
\]

\begin{Ass}[Local exponential family model]\label{A1}
The parametric family $\mathcal{P} = \lbrace \mathbb{P}_{\theta} \rbrace_{\theta \in \Theta}$ in Notation~\ref{not:setting} is an exponential family.
More precisely, there are two functions $C,B \in C^2\left( \Theta, \mathbb{R} \right)$, a non-negative function $p: \mathcal{Y} \to [0, \infty)$, and a sufficient statistic $T: \mathcal{Y} \to \mathbb{R}$ such that
\[
	p(y, \theta) := d \mathbb{P}_{\theta} / d \mathbb{P} (y) = p(y) \exp \left[ T(y) C(\theta) - B(\theta) \right], \qquad \theta \in \Theta,
\]
where~$C$ is strictly monotonic increasing. 
The parameter~$\theta$ satisfies $B'(\theta) = \theta \, C'(\theta)$,
\begin{equation}\label{eq:A1}
	\int p(y, \theta) \mathbb{P}(dy) = 1, \quad \text{ and } \quad 
	\mathbb{E}_{\theta} \left[ T(Y) \right] = \int T(y) p(y, \theta) \mathbb{P}(dy) = \theta.
\end{equation}
\end{Ass}

We recall the notions of the Fisher information
\[
	I (\theta) :=  - \mathbb{E} \left[ \frac{ \partial^2 }{ \partial \theta^2 } \log p(y, \theta) \right], \quad \theta \in \Theta,
\] 
and of the Kullback--Leibler divergence
\[
	\mathcal{KL}(\theta,\theta') 
	:=\mathcal{KL}\left( \mathbb{P}_{\theta},\mathbb{P}_{\theta'} \right)  
	:= \int \ln \left( \frac{ d( \mathbb{P}_{\theta} ) }{ d( \mathbb{P}_{\theta'} ) } \right) \mathbb{P}_{\theta} (dy), \quad \theta, \theta' \in \Theta.
\]

The Propagation-Separation Approach estimates iteratively the unknown parameter function~$\theta(.)$. 
Here, we consider its simplest version, supposing the parameter function~$\theta(.)$ to be piecewise constant with sharp discontinuities.
The pointwise estimator equals a weighted mean of the observations. 
In each iteration step~$k$ the adaptive weights are readjusted using the previously aggregated information. 
More precisely, the adaptive weights are defined as a product of two kernels. 
The \emph{location kernel} describes for every point $X_i \in \mathcal{X}$ the increasing neighborhood~$U_i^{(k)} \subseteq \mathcal{X}$ under consideration, leading to an advancing variance reduction. 
The \emph{adaptation kernel} uses the Kullback-Leibler divergence for a comparison of the pointwise parameter estimates from the previous iteration step.
This avoids blurring at structural borders.
An additional \emph{memory step} ensures a certain stability of estimates.
In each iteration step, the \emph{memory penalty} compares, for every design point, the new estimate with the previous one. 
In case of a significant difference, the new estimate is relaxed, replacing it by a value between the two estimates. 
The memory step provides a smooth transition of the pointwise estimates during iteration. 
We emphasize that the Propagation-Separation Approach does not use adaptive parameters. 
It is adaptive in the sense that the returned estimator function is based on structure-adaptive weights that describe the homogeneity regions of the unknown parameter function~$\theta(.)$. 
See Algorithm~\ref{algorithmMS}, stated below, for a formal description and \citet{PoSp05} and \citet{BecMat2013} for more details.

\begin{Not}\label{not:algorithm}\hspace{1pt}
Suppose Assumption~\ref{A1}.
We fix three non-increasing kernel functions
\[
	K_{\mathrm{loc}}, K_{\mathrm{ad}}, K_{\mathrm{me}}: [0, \infty) \to [0,1]
\]
with support~$[0,1)$, satisfying $K_{\cdot}(0) = 1$. 
These kernels will be used for location, for adaptation, and for the memory step, respectively.
Moreover, let $\lambda>0$ denote the bandwidth of the adaptation kernel, and let~$\lbrace h^{(k)} \rbrace_{k=0}^{k^*}$ be an increasing sequence of pre-specified location bandwidths with $h^{(0)} > 0$.
For the memory step, we choose the minimal memory effect $\eta_0 \in [0,1)$ and the memory bandwidth $\tau > 0$.
Then, we call the weighted mean
\begin{equation}\label{eq:non-adEst}
	\overline{\theta}_i^{(k)} := \sum_{j=1}^n \overline{w}_{ij}^{(k)} T(Y_j) / \overline{N}_i^{(k)},
\end{equation}
 the non-adaptive estimator of~$\theta_i$, where $\overline{w}_{ij}^{(k)} := K_{\mathrm{loc}} \left( \delta(X_i, X_j) / h^{(k)} \right)$ and $\overline{N}_i^{(k)} := \sum_j \overline{w}_{ij}^{(k)}$.
\end{Not}

\begin{Alg}[The original Propagation-Separation Algorithm]\label{algorithmMS}\hspace{1 pt}
\begin{compactenum}
\item Input parameters: Sequence of bandwidths~$\lbrace h^{(k)} \rbrace_{k=0}^{k^*}$, adaptation bandwidth~$\lambda$,\\
the memory bandwidth~$\tau$, and the minimal memory effect~$\eta_0$.
\item Initialization: 
$\hat{\theta}_i^{(0)} := \overline{\theta}_i^{(0)}$ and $\hat{N}_i^{(0)} := \overline{N}_i^{(0)}$ for all $i \in \lbrace 1,...,n \rbrace$,~$k:= 1$. 
\item Iteration: Calculate, for every $i,j = 1,...,n$, \\
the non-adaptive weights $\overline{w}_{ij}^{(k)} := K_{\mathrm{loc}} \left( \delta(X_i, X_j) / h^{(k)} \right), \\$
the statistical penalty $s_{ij}^{(k)} := \hat{N}_i^{(k-1)} \mathcal{KL}(\hat{\theta}_i^{(k-1)}, \hat{\theta}_j^{(k-1)})$, \\
the adaptive weights $\tilde{w}_{ij}^{(k)} := \overline{w}_{ij}^{(k)} \cdot K_{\mathrm{ad}} \left( s_{ij}^{(k)} / \lambda \right)$, \\
the sum of the adaptive weights $\tilde{N}_i^{(k)} := \sum_j \tilde{w}_{ij}^{(k)}$, \\
and the adaptive estimator
\[
	\tilde{\theta}_i^{(k)} := \sum_{j=1}^n \tilde{w}_{ij}^{(k)} T(Y_j) / \tilde{N}_i^{(k)}.
\]
\item Memory step: Calculate, for every $i,j = 1,...,n$, \\
the sum of the non-adaptive weights $\overline{N}_i^{(k)} := \sum_j \overline{w}_{ij}^{(k)}$, \\
the memory penalty $m_i^{(k)} := \overline{N}_i^{(k)} \mathcal{KL}(\tilde{\theta}_i^{(k)}, \hat{\theta}_i^{(k-1)})$, \\
the relaxation weight $\eta_i^{(k)} := (1 - \eta_0) K_{\mathrm{me}} \left( m_i^{(k)} / \tau \right)$, \\
the relaxed estimator
\[
	\hat{\theta}_i^{(k)} := \eta_i^{(k)} \tilde{\theta}_i^{(k)} + (1 - \eta_i^{(k)}) \hat{\theta}_i^{(k-1)},
\]
and the relaxed sum of the adaptive weights $\hat{N}_i^{(k)} := \eta_i^{(k)} \tilde{N}_i^{(k)} + (1 - \eta_i^{(k)}) \hat{N}_i^{(k-1)}$.
\item Stopping: Stop if~$k = k^*$, and return~$\hat{\theta}_i^{(k^*)}$ for all $i \in \lbrace 1,...,n \rbrace$, \\
otherwise increase~$k$ by~$1$.
\end{compactenum}
\end{Alg}

We emphasize that the the data-driven statistical penalty~$s_{ij}^{(k)}$ makes the adaptive weights~$\tilde{w}_{ij}^{(k)}$, their sum~$\tilde{N}_i^{(k)}$, and the relaxed sum~$\hat{N}_i^{(k)}$ random.
In contrast, we notice that the input parameters, the non-adaptive weights~$\overline{w}_{ij}^{(k)}$, and their sum~$\overline{N}_i^{(k)}$ are deterministic.
Here, we concentrate on a simplified procedure, where the memory step is omitted.

\begin{Not}[Simplified algorithm]\label{algorithm}
In the rest of this article, we refer to Algorithm~\ref{algorithmMS} as the \emph{original algorithm} with aggregated estimates~$\lbrace \hat{\theta}_i^{(k)} \rbrace_{i,k}$.
The formal choice $\eta_i^{(k)} := 1$ for every $i \in \lbrace 1,...,n \rbrace$ and each $k \in \lbrace 1,..., k^* \rbrace$ omits the memory step, leading to the \emph{simplified algorithm} with adaptive estimates~$\lbrace \tilde{\theta}_i^{(k)} \rbrace_{i,k}$.
\end{Not}

The choice of the input parameters is crucial for the behavior of the algorithm.
Since the initial estimator~$\tilde{\theta}_i^{(0)}$ is non-adaptive the corresponding location bandwidth~$h^{(0)}$ should be small. 
A choice of~$h^{(0)}$ such that $\overline{w}_{ij}^{(0)} = 0$ for all $X_i \neq X_j$ avoids blurring at the boundaries of the homogeneity regions. 
The subsequent bandwidths~$\lbrace h^{(k)} \rbrace_{k=1}^{k^*}$ should be increasing. 
For instance, they may ensure a constant variance reduction of the estimator \citep{poas} or an exponential growth \citep{PoSp05} of the mean number of design points $X_j \in \mathcal{X}$ with non-zero weights $\overline{w}_{ij}^{(k)} \neq 0$, $X_i \in \mathcal{X}$. 
The maximal number of iterations~$k^*$ specifies the maximal location bandwidth~$h^{(k^*)}$. 
This is mainly bounded by the available computation time. 
However, in the case of model misspecification the resulting estimation bias can be reduced by an accurate stopping criterion as we will discuss in \S~\ref{sec:PS-future}.

The adaptation bandwidth~$\lambda$ specifies the amount of adaptation. 
For $\lambda \to \infty$ the algorithm results in non-adaptive estimates as defined in Equation~\eqref{eq:non-adEst} (over-smoothing), while small values lead to adaptation to noise (under-smoothing).
\citet[\S 3.5]{PoSp05} introduced a choice of~$\lambda$ by a strategy, called the \emph{propagation condition}. 
This recommends to use the smallest value for~$\lambda$ that provides under homogeneity a similar behavior as non-adaptive smoothing. 
We use a revised formulation that was introduced in \citep[\S 2.3]{BecMat2013}.
It is based on the function~$\mathfrak{Z}_{\lambda}: \lbrace 0,...,k^* \rbrace \times (0, 1) \times \Theta \times \lbrace 1,...,n \rbrace \to [0, \infty)$ given by
\[
	\mathfrak{Z}_{\lambda}(k, p; \theta, i ) := \inf \left\{ z > 0: \mathbb{P} \left( \overline{N}_i^{(k)} \mathcal{KL}(\tilde{\theta}_i^{(k)}(\lambda), \theta) > z \right) \leq p \right\},
\]
where $\lambda > 0$ is fixed.
Here, $\tilde{\theta}_i^{(k)}(\lambda)$ denotes the adaptive estimator in the position $X_i \in \mathcal{X}$,
resulting from the simplified algorithm in Notation~\ref{algorithm} with the adaptation bandwidth~$\lambda$ and observations $Y_j \overset{\text{iid}}{\sim} \mathbb{P}_{\theta}$ for all $j \in \lbrace 1,...,n \rbrace$ with~$\theta(.) \equiv \theta$.

\begin{Def}[Homogeneous propagation condition]\label{def:propCond}
We say that the adaptation bandwidth $\lambda > 0$ is chosen in accordance with the homogeneous propagation condition at level $\epsilon > 0$ for $\theta \in \Theta$ if the function~$\mathfrak{Z}_{\lambda}(., p; \theta, i)$ is non-increasing for all $p \in (\epsilon, 1)$ and every $i \in \lbrace 1,...,n \rbrace$.
\end{Def}

The study in \citep[\S 4.1]{BecMat2013} points out that the choice of~$\lambda$ by the homogeneous propagation condition is invariant with respect to the underlying parameter~$\theta$ for the Gaussian and the exponential distribution and, as a consequence, for the log-normal, Rayleigh, Weibull, and Pareto distributions. 
Else, some parameter~$\theta^*$ which yields a sufficiently large choice should be identified, such that the homogeneous propagation condition holds for all (unknown) parameters~$\theta_i$ with $i \in \lbrace 1,..., n \rbrace$ as well. 
Hence, the homogeneous propagation condition allows a choice of~$\lambda$ by simulations and hence independent of the data at hand. 
We refer the reader to \citet[\S 3.4 \& 3.5]{PoSp05} and \citet[\S 2.3 \& 4.1]{BecMat2013} for more details.

\section{Theoretical properties under model misspecification}\label{sec:Theory}

We consider the local exponential family model in Assumption~\ref{A1} (page~\pageref{A1}) and the simplified procedure in Notation~\ref{algorithm}.
In \citep[\S 2.4]{BecMat2013}, its general behavior was illustrated on two examples. There, we observed the following. 
For piecewise constant parameter functions~$\theta(.)$, the algorithm detects sufficiently sharp discontinuities providing a consistent estimation function. 
For small discontinuities this separation property fails. 
This leads to a bounded estimation bias since the algorithm treats non-separated homogeneity regions as one yielding similar results as non-adaptive smoothing. 
For piecewise smooth parameter functions~$\theta(.)$ the algorithm results in a step function which is mainly determined by the local smoothness of the parameter function~$\theta(.)$ and the adaptation bandwidth~$\lambda$. 
An appropriate stopping criterion may reduce the corresponding estimation bias. 

Our theoretical study in \citep{BecMat2013} focused on piecewise constant functions with sharp discontinuities. 
In this article, we aim to verify the mentioned heuristic observations for the case of model misspecification. 
First, we will generalize the propagation condition for the choice of the adaptation bandwidth to inhomogeneous settings with bounded variability. 
Then, the propagation and the stability property will follow for parameter functions with (piecewise) bounded variability in a similar manner as under (piecewise) homogeneity, see \citet[\S 3.1 \& 3.3]{BecMat2013}. 
Furthermore, we will introduce a specific step function which approximates the adaptive estimates, resulting from the simplified Propagation-Separation Approach.

\subsection{Inhomogeneous propagation condition}\label{sec:inhomPropCond}

The homogeneous propagation condition in Definition \ref{def:propCond} bounds the probability of adaptation to noise, supposing a constant parameter function. 
In \citep[\S 3]{PoSp05}, this was used to verify propagation and a certain stability of estimates for (piecewise) constant parameter functions. 
In order to extend these properties to (piecewise) bounded parameter functions, we will formulate an inhomogeneous propagation condition.
Like before under homogeneity, we will consider an artificial data set.
Then, we aim to ensure a similar behavior of the algorithm as for non-adaptive estimation for every locally varying function which satisfies a pre-specified variability bound.

Our inhomogeneous propagation condition is motivated by Theorem~\ref{thm:PS22} in Appendix~\ref{app:auxiliary}.
This can be considered as the inhomogeneous analog of Theorem~\ref{thm:PS21}, where the homogeneous propagation condition was based on. 
For the non-adaptive estimator, Theorem~\ref{thm:PS22} establishes the exponential bound 
\[
	\mathbb{P} (\overline{N}_i^{(k)} \mathcal{KL}(\overline{\theta}_i^{(k)}, \mathbb{E} \overline{\theta}_i^{(k)}) > z ) \leq 2 e^{-z/\varkappa^2} + \breve{p}_{\varkappa}
\] 
for all $z > 0$, $i \in \lbrace 1,...,n \rbrace$, and $k \in \lbrace 0,..., k^* \rbrace$, where we refer the reader to the Notations~\ref{not:varkappa} and~\ref{not:pKappa2} for the definitions of $\varkappa \geq 1$ and~$\breve{p}_{\varkappa} \in [0,1]$. 
This result implies that the Kullback-Leibler divergence $\mathcal{KL} (\overline{\theta}_i^{(k)}, \mathbb{E} \overline{\theta}_i^{(k)})$ decreases, in probability, at least with rate~$ \overline{N}_i^{(k)}$. 
We observe that
\[
	\mathbb{E} \overline{\theta}_i^{(k)} 
	= \sum_{j=1}^n \mathbb{E} \left[ \overline{w}_{ij}^{(k)} T(Y_j) / \overline{N}_i^{(k)} \right]
	= \sum_{j=1}^n \overline{w}_{ij}^{(k)} \theta_j / \overline{N}_i^{(k)},
\]
whereas
\[
	\mathbb{E} \tilde{\theta}_i^{(k)}
	= \sum_{j=1}^n \mathbb{E} \left[ \tilde{w}_{ij}^{(k)} T(Y_j) / \tilde{N}_i^{(k)} \right]
	\neq \sum_{j=1}^n \tilde{w}_{ij}^{(k)} \theta_j / \tilde{N}_i^{(k)}.
\]
Both sums can be considered as an adaptive analog of~$\mathbb{E} \overline{\theta}_i^{(k)}$. 
Since the latter is much easier to compute, we concentrate thereon.
Recall that the adaptive weights and their sum are random.

\begin{Not}\label{not:mathcalE}
We set
\[
	\mathcal{E} \tilde{\theta}_i^{(k)}
	:= \sum_{j=1}^n \tilde{w}_{ij}^{(k)} \theta_j / \tilde{N}_i^{(k)}. 
\]
\end{Not}

Next we specify the considered inhomogeneous setting.
Following \citet[\S 5.2]{PoSp05}, we presume that the variability of the parameter function~$\theta(.)$ is smaller in order than the rate of convergence~$\overline{N}_i^{(k)}$ in Theorem~\ref{thm:PS22}. 
Here, we even require the rate~$\max_{j'} \overline{N}_{j'}^{(k)}$ in order to ensure that $\overline{N}_i^{(k)} / \max_{j'} \overline{N}_{j'}^{(k)} \leq 1$ for every $i \in \lbrace 1,...,n \rbrace$.
More precisely, we require the existence of a constant $\varphi \geq 0$ such that 
\begin{equation}\label{eq:inhomBound}
	\mathcal{KL} \left( \theta_i, \theta_j \right) \leq \varphi^2 / \max_{j'} \overline{N}_{j'}^{(k)} \quad \text{ for all } 
	X_j \in U_i^{(k)} := \lbrace X_j \in \mathcal{X}: \overline{w}_{ij}^{(k)} > 0 \rbrace
\end{equation}
for every $i \in \lbrace 1,..., n \rbrace$ and each  $k \in \lbrace 0,..., k^* \rbrace$.
In this subsection, we require Equation~\eqref{eq:inhomBound} with $k := k^*$ for all $i,j \in \lbrace 1,...,n \rbrace$, but, in \S~\ref{sec:propagation}, we will only consider the points in a certain neighborhood, for instance all $X_j \in U_i^{(k)}$ with $k \in \lbrace 0,...,k^* \rbrace$. 
For brevity, we denote $\varphi_0 := \varphi / \max_i  (\overline{N}_i^{(k^*)} )^{1/2}$.

We proceed as under homogeneity, see Definition~\ref{def:propCond} for comparison. 
 
\begin{Not}\label{not:zLambdaInhom}
For every~$\lambda > 0$, we consider the function
\[ 
	\hat{\mathfrak{Z}}_{\lambda}: \lbrace 0,..., k^* \rbrace \times (0,1) \times \Theta^n \times \lbrace 1,...,n \rbrace \to [0, \infty) 
\]
defined by
\[
	\hat{\mathfrak{Z}}_{\lambda}(k, p; \lbrace \theta_j \rbrace_{j=1}^n, i ) := \inf \left\{ z > 0: \mathbb{P} \left( \overline{N}_i^{(k)} \mathcal{KL}(\tilde{\theta}_i^{(k)}(\lambda), \mathcal{E} \tilde{\theta}_i^{(k)}(\lambda)) > z \right) \leq p \right\}, 
\]
where~$\mathcal{E} \tilde{\theta}_i^{(k)}$ is as in Notation~\ref{not:mathcalE}, and~$\tilde{\theta}_i^{(k)}(\lambda)$ denotes the adaptive estimator in the position $X_i \in \mathcal{X}$,
resulting from the simplified algorithm in Notation~\ref{algorithm} with the adaptation bandwidth~$\lambda$ and the statistically independent observations $Y_j \sim \mathbb{P}_{\theta_j} \in \mathcal{P}$, $j \in \lbrace 1,...,n \rbrace$.
\end{Not}

In order to enable the application of Equation~\eqref{eq:varkappa} and Lemma~\ref{lem:PS52}, we restrict the range of the parameter function~$\theta(.)$. 
Thus, we introduce a subset $\Theta^* \subseteq \Theta$ with $\lbrace \theta_j \rbrace_{j=1}^n \in (\Theta^*)^n$.

\begin{Def}[Inhomogeneous propagation condition]\label{def:inhomPropCond}
Let $\epsilon > 0$ and $\varphi_0 \geq 0$ be constants. 
The adaptation bandwidth~$\lambda > 0$ satisfies the inhomogeneous propagation condition at probability level~$\epsilon$ and variability level $\varphi_0$ for the parameter set~$\Theta^* \subseteq \Theta$ if the function~$\hat{\mathfrak{Z}}_{\lambda}(., p; \lbrace \theta_j \rbrace_{j=1}^n, i)$ is non-increasing for all $p \in (\epsilon, 1)$, every $i \in \lbrace 1,...,n \rbrace$, and each parameter function~$\theta(.)$ with $\lbrace \theta_j \rbrace_{j=1}^n \in (\Theta^*)^n$ and $\mathcal{KL} (\theta_j, \theta_{j'}) \leq \varphi_0^2$ for all $j, j' \in \lbrace 1,...,n \rbrace$.
\end{Def}

\begin{Rem}
For $\varphi_0 := 0$ and $\Theta^* := \lbrace \theta \rbrace$, the inhomogeneous propagation condition equals the homogeneous propagation condition in Definition~\ref{def:propCond}. 
\end{Rem}

\subsection{Locally varying parameter functions with sharp discontinuities} \label{sec:propagation}

We deduce from the inhomogeneous propagation condition and Theorem~\ref{thm:PS22} in Appendix~\ref{app:auxiliary} an exponential bound for the probability $\mathbb{P} ( \overline{N}_i^{(k)} \mathcal{KL}(\tilde{\theta}_i^{(k)}, \theta_i) > z )$ 
of the Kullback-Leibler divergence between the adaptive estimator~$ \tilde{\theta}_i^{(k)}$ and its true parameter~$\theta_i$ to exceed the error bound~$z / \overline{N}_i^{(k)}$. 
The following proposition provides the inhomogeneous analog of \citep[Proposition 3.1]{BecMat2013}.
It requires slightly different assumptions and yields a different exponent in the exponential bound of the excess probability.

\begin{Prop}[Propagation and stability under bounded variability]\label{prop:propagation1}
Let Assumption~\ref{A1} (page~\pageref{A1}) be fulfilled, and let the adaptation bandwidth~$\lambda$ be chosen in accordance with the inhomogeneous propagation condition at probability level~$\epsilon>0$ and variability level $\varphi_0 > 0$ for some set $\Theta^* \subseteq \Theta$ satisfying $\lbrace \theta_i \rbrace_{i=1}^n \in (\Theta^*)^n$.
Additionally, we recall Notation~\ref{not:varkappa} in Appendix~\ref{app:auxiliary}, and we choose $\varkappa \geq 1$ sufficiently large such that $\Theta^* \subseteq \Theta_{\varkappa}$.
If $\mathcal{KL} \left( \theta_{i}, \theta_j \right) \leq \varphi^2 / \max_{j'} \overline{N}_{j'}^{(k_0)} = \varphi_0^2$ holds for all $ i,j \in \lbrace 1,...,n \rbrace$ and some fixed $k_0 \in \lbrace 1,..., k^* \rbrace$, then we get 
\begin{equation}\label{eq:propCondInhom1}
	\mathbb{P} \left( \overline{N}_i^{(k)} \mathcal{KL} \left( \tilde{\theta}_i^{(k)}, \theta_i \right) > z \right) 
	\leq \max \left\{ 2 e^{- \left[ \sqrt{z} / \varkappa - \varphi \right]^2 / \varkappa^2}, \epsilon \right\} + \breve{p}_{\varkappa,0}
\end{equation}
for each $i \in \lbrace 1,...,n \rbrace$, $k \in \lbrace 0,...,k_0 \rbrace$, and all~$z > \varkappa^2 \varphi^2$, where~$\breve{p}_{\varkappa,0}$ is as in Notation~\ref{not:pKappa2}. 
In particular, for all~$k_1 \leq k_2 \leq k_0$, it holds
\begin{equation}\label{eq:propCondInhom2}
	\mathbb{P} \left( \overline{N}_i^{(k_2)} \mathcal{KL} \left( \tilde{\theta}_i^{(k_2)}, \mathcal{E} \tilde{\theta}_i^{(k_2)} \right) > z \right) \leq \max \left\{ \mathbb{P} \left( \overline{N}_i^{(k_1)} \mathcal{KL} \left( \tilde{\theta}_i^{(k_1)}, \mathcal{E} \tilde{\theta}_i^{(k_1)} \right) > z \right), \epsilon \right\},
\end{equation}
where~$\mathcal{E} \tilde{\theta}_i^{(k)} = \sum_j \tilde{w}_{ij}^{(k)} \theta_j / \tilde{N}_i^{(k)}$ as in Notation~\ref{not:mathcalE}.
\end{Prop}

Next we consider piecewise bounded functions with sharp discontinuities.
We recall some auxiliary notations.

\begin{Not}\label{not:connect}\hspace{1 pt}
For any set~$M$, we define
\[
	\mathfrak{C}(M) := \bigcap \, \left\{ M_c: M_c \text{ is a connected space and } M \subseteq M_c 
	\right\}.
\]
Then, we call the discrete set $M := \lbrace X_{j} \rbrace_{j=1}^m \subseteq \mathcal{X}$ convex if~$\mathfrak{C}(M)$ is convex and
\[
	X_j \in M \text{ if and only if } X_j \in \mathfrak{C}(M),  \qquad X_j \in \mathcal{X}. 
\]
\end{Not}

Then, the setting is described by the following structural assumption.

\begin{Ass}\label{ASinhom}
Suppose the existence of a non-trivial partition $\mathcal{V} := \lbrace \mathcal{V}_i \rbrace_i$ of~$\mathcal{X}$ such that, for every $X_{i} \in \mathcal{X}$, there are constants~$\phi_{i} > \varphi_0 \geq 0$ and a convex 
neighborhood $\mathcal{V}_{i} \subseteq \mathcal{X}$ which satisfy
\[
\begin{cases}
	\mathcal{KL} \left( \theta_{i}, \theta_j \right) \leq \varphi_0^2 &\text{ for all } X_j \in \mathcal{V}_{i},\\ 
	\mathcal{KL} \left( \theta_{i}, \theta_j \right) > \phi_{i}^2 &\text{ for all } X_j \notin \mathcal{V}_{i}.
\end{cases}
\]
\end{Ass}

We recall some notations from \citep{BecMat2013}.
The effective sample size concentrates on the case where the considered neighborhood~$U_i^{(k)} = \lbrace X_j \in \mathcal{X}: \overline{w}_{ij}^{(k)} > 0 \rbrace$ is larger than the corresponding region region~$\mathcal{V}_i$. 

\begin{Not}\label{not:effectiveSample}
We define, for each $i \in \lbrace 1,...,n \rbrace$ and every $ k \in \lbrace 0,..., k^* \rbrace $, the effective sample size and its local minimum
\begin{equation}\label{eq:effectiveSample}
	\overline{n}_i^{(k)} := 
	\sum_{X_{j} \in \mathcal{V}_{i} \cap U_{i}^{(k)}} \overline{w}_{ij}^{(k)}
	\qquad \text{ and } \qquad
	n_i^{(k)} := \underset{X_{j} \in U_i^{(k)} }{\min} \overline{n}_j^{(k)}.
\end{equation}
\end{Not}

As it turns out, the quantities~$n_i^{(k)}$ determine a lower bound for the
stepsizes~$\phi_{i}$ which allows the detection of the associated discontinuity by the algorithm.
In the following theorem, we consider two events.
On the first one, $\mathcal{B}^{(k)} (z)$, the estimation error is bounded from above, and on the second one, $M^{(k)}(z)$, the discontinuities are sufficiently sharp for separation, see Proposition~\ref{prop:separationMS} in Appendix~\ref{app:auxiliary}.

\begin{Not}
Let the constants $\phi_{i} > 0$, $i \in \lbrace 1,..., n \rbrace$, be as in Assumption~\ref{ASinhom}, and fix $\lambda > 0$ and $z>0$.
Additionally, we recall Notation~\ref{not:varkappa} and choose $\varkappa \geq 1$ such that $\lbrace \theta_i \rbrace_{i=1}^n \in (\Theta_{\varkappa})^n$.
Then, we set
\begin{equation}\label{eq:eventBk}
	\mathcal{B}^{(k)} (z) := \bigcap_{i=1}^n \left\{ \overline{n}_i^{(k)} \mathcal{KL} ( \tilde{\theta}_i^{(k)}, \theta_i ) \leq z \right\}, \qquad k \in \lbrace 0,...,k^* \rbrace,
\end{equation}
$M^{(0)}(z) := \Omega$, and
\begin{equation}\label{eq:varphi3}
	M^{(k)}(z) := 
	\bigcap_{k'=0}^{k-1} \bigcap_{i=1}^n \left\{ \phi_{i} > \varkappa \left[ \sqrt{\lambda/ \tilde{N}_{i}^{(k')} } + 2 \sqrt{ z/ n_i^{(k')} } \right] \right\}, \quad k \in \lbrace 1,...,k^* \rbrace.
\end{equation}
\end{Not}

\begin{Thm}[Propagation property under piecewise boundedness]\label{thm:propagation2}
Suppose Assumptions~\ref{A1} (page~\pageref{A1}) and~\ref{ASinhom} to be satisfied. 
Additionally, let the adaptation bandwidth~$\lambda$, 
the constant $\varkappa \geq 1$, and the corresponding set $\Theta_{\varkappa} \subseteq \Theta$ be as in Proposition~\ref{prop:propagation1} and $h^{(0)} > 0$ sufficiently small such that $\overline{w}_{ij}^{(k)} = 0$ for all $X_i \neq X_j$.
Finally, we fix some iteration step $k_0 \in \lbrace 0,..., k^* \rbrace$ and some constant $\varphi \geq 0$ such that $\varphi^2 / \max_i \overline{n}_i^{(k_0)} = \varphi_0^2$.
If $z > \varkappa^2 \varphi^2$ satisfies $\mathbb{P} \left( M^{(k_0)}(z) \right) > 0$, then it holds
\begin{eqnarray*}
	\mathbb{P} \left( \mathcal{B}^{(k_0)}(z) \vert M^{(k_0)}(z) \right) 
	\geq 1 - \frac{ \breve{p}_{\varkappa,0} + (k_0+1) \, \max \left\{ 2 n e^{- \left[ \sqrt{z} / \varkappa - \varphi \right]^2 / \varkappa^2}, n \epsilon \right\} }{ \mathbb{P} \left( M^{(k_0)}(z) \right)},
\end{eqnarray*}
where~$\breve{p}_{\varkappa,0}$ is as in Notation~\ref{not:pKappa2}.
\end{Thm}

\begin{Rem}
Theorem~\ref{thm:propagation2} yields a meaningful result for
$z \geq \varkappa^2 [ \varkappa \sqrt{q \log(n)} + \varphi ]^2$ 
and small values of~$\epsilon$ or, at best, $\epsilon := c_{\epsilon} n^{-q}$
with~$q > 1$ and~$c_{\epsilon} > 0$.
\end{Rem}

\subsection{Consequences of a violated structural assumption}\label{sec:stability}

The previous results only hold for parameter functions with sharp discontinuities. 
What happens in the case of a violated structural assumption?
In \citep[\S 2.4]{BecMat2013}, we observed for simulated examples with Gaussian distributed observations that the estimation function resulted in a step function in the case of a piecewise constant parameter function and as well for a piecewise smooth function.
Therefore, we will introduce a specific step function, that we will call the \emph{associated step function} of the Propagation-Separation Approach.
Then, we will establish an upper bound for the pointwise Kullback-Leibler divergence between the adaptive estimator of the simplified procedure in Notation~\ref{algorithm} and the corresponding value of the associated step function.

Applying the Propa\-gation-Sepa\-ration Approach with some fixed adaptation bandwidth $\lambda > 0$ provides, for every $k \in \lbrace 1,..., k^* \rbrace$, a set of adaptive weights $\lbrace \tilde{w}_{ij}^{(k)} \rbrace_{i,j=1}^n$.
In particular, for $k \in \lbrace 0,..., k^* \rbrace$, this yields the subsets
\begin{equation}\label{eq:partitionHik}
	\mathcal{H}_i^{(k)} := \left\{ \mathcal X_j \in \mathcal{X}: \tilde{w}_{il}^{(k+1)} > 0 \text{ if and only if } \tilde{w}_{jl}^{(k+1)} > 0 \text{ for all } X_l \in \mathcal{X} \right\},
\end{equation}
where we set $\tilde{w}_{ij}^{(k^* + 1)} := \overline{w}_{ij}^{(k^*)} \cdot K_{\mathrm{ad}} (s_{ij}^{(k^*)} / \lambda)$.
They are based on an equivalence relation, yielding, for every parameter function~$\theta(.)$, a well-defined partition~$\lbrace H_l^{(k)} \rbrace_{l=1}^m$ of the design space~$\mathcal{X}$ into $m \leq n$ regions.
We introduce a step function whose steps match this partition~$\lbrace H_l^{(k)} \rbrace_{l=1}^m$.

\begin{Def}\label{def:assStepf}
Let~$\boldsymbol{1}$ denote the indicator function, and let~$\theta_l^{(k)}$ be the mean value of the~$n_l$ estimates~$\tilde{\theta}_{l_j}^{(k)}$ corresponding to the design points~$\lbrace X_{l_j} \rbrace_{j=1}^{n_l}$ which form the region~$ H_l^{(k)}$. 
Then, we call the piecewise constant function
\begin{equation}\label{eq:stepF1}
	\breve{\theta}^{(k)} (X_i) := \sum_{l=1}^m \theta_l^{(k)} \boldsymbol{1}_{H_l^{(k)}} (X_i)
	\quad \text{ with } \quad
	\theta_l^{(k)} := \frac{1}{n_l} \sum_{j=1}^{n_l} \tilde{\theta}_{l_j}^{(k)}
\end{equation}
the associated step function of~$\theta(.)$ in step~$k$.
For $i \in \lbrace 1,...,n \rbrace$ and $k \in \lbrace 1,...,k^* \rbrace$, we set $\breve{\theta}^{(k)}_i := \breve{\theta}^{(k)}(X_i)$.
\end{Def}

The associated step function satisfies the following property.

\begin{Lem}\label{lem:Hik}
For all $i \in \lbrace 1,...,n \rbrace$ and $k \in \lbrace 0,..., k^* \rbrace$, it holds
\begin{equation}\label{eq:Hik}
	\mathcal{KL} \left( \breve{\theta}^{(k)}_i, \tilde{\theta}_i^{(k)} \right) \leq \max \lbrace \lambda / \tilde{N}_j^{(k)}: X_j \in \mathcal{H}_i^{(k)} \rbrace.
\end{equation}
\end{Lem}

\begin{proof}
We know from Lemma~\ref{lem:A1} that the Kullback-Leibler divergence is convex with respect to the first argument. 
Therefore, it holds
\[
	\mathcal{KL} \left( \breve{\theta}^{(k)}_i, \tilde{\theta}_i^{(k)} \right)
	\leq \max \left\{ \mathcal{KL} \left( \tilde{\theta}_j^{(k)}, \tilde{\theta}_i^{(k)} \right): X_j \in \mathcal{H}_i^{(k)} \right\}.
\]
Since $X_j \in \mathcal{H}_i^{(k)}$ implies $\tilde{w}_{ji}^{(k+1)} > 0$, we have $\mathcal{KL} \left( \tilde{\theta}_j^{(k)}, \tilde{\theta}_i^{(k)} \right) \leq \lambda / \tilde{N}_j^{(k)}$, which leads to the assertion. 
\end{proof}

In \S~\ref{sec:simStep}, we will illustrate the formation of the associated step function during iteration. 
The corresponding simulations suggest its immutability for sufficiently large bandwidths. 
Additionally, we will see that, in the presented examples, the sets $\lbrace \mathcal X_j \in \mathcal{X}: \tilde{w}_{ij}^{(k)} > 0 \rbrace$ with $i \in \lbrace 1,...,n \rbrace$ 
form a well-defined partition of the design space~$\mathcal{X}$ if~$k$ is sufficiently large.
However, both heuristic observations could not be theoretically justified for reasons that we will discuss in~\S~\ref{sec:convergence}.

\section{The inhomogeneous propagation condition in practice}\label{sec:inhomPCPractice}

In this section, we will present further details concerning the practical application of the inhomogeneous propagation condition. 
In \cite[\S 4.2]{BecMat2013}, we explained how the homogeneous propagation condition can be applied in practice. 
In contrast, the inhomogeneous propagation condition cannot be applied directly if $\varphi_0 > 0$. 
Here, we need to ensure that the criterion is fulfilled for \emph{every} parameter function satisfying $\lbrace \theta_i \rbrace_{i=1}^n \in (\Theta^*)^n$ and $\mathcal{KL} (\theta_i, \theta_j) \leq \varphi_0^2$ for all $i,j \in \lbrace 1,...,n \rbrace$. 
Therefore, we recommend to choose some $\lambda > 0$ in accordance with the homogeneous propagation condition and to increase it such that the inhomogeneous propagation condition holds as well. 
Apart from the Gaussian and log-normal distribution, the practical use of our precise choice is questionable due to the size of the involved constants. 
Nevertheless, it suggests the \emph{existence} of an appropriate value. 
Hence, the inhomogeneous propagation condition is in the first instance of theoretical interest. 
It allows the desired extension of the propagation and the stability property to (piecewise) bounded functions. 
The justification of our choice will be based on a comparison of the homogeneous and the inhomogeneous propagation condition. 
In order to avoid confusion, we introduce the following notation.

\begin{Not}\label{not:inhomVsHom}
Let the parametric family~$\mathcal{P}$ satisfy Assumption~\ref{A1} (page~\pageref{A1}) with a strictly monotonic sufficient statistic~$T$.  
We fix some constant $\varphi_0 > 0$ and a subset $\Theta^* \subseteq \Theta$. 
Then, we consider two data sets~$ \lbrace (X_i, Y_i) \rbrace_{i=1}^n$ and~$ \lbrace (X_i, \mathcal{Y}_i) \rbrace_{i=1}^n$, where
\begin{compactitem}
\item $Y_i \sim \mathbb{P}_{\theta_i} \in \mathcal{P}$ with $\lbrace \theta_i \rbrace_{i=1}^n \in (\Theta^*)^n$ and $\mathcal{KL} (\theta_i, \theta_j) \leq \varphi_0^2$ for all $i,j \in \lbrace 1,...,n \rbrace$,
\item $\mathcal{Y}_i \sim \mathbb{P}_{\vartheta_i} \in \mathcal{P}$ with $\vartheta_i \equiv \vartheta$ for some $\vartheta \in \Theta^*$ (homogeneity).
\end{compactitem}
\end{Not}

In the rest of this section, we will write~$\mathcal{Y}$ and~$\vartheta$ whenever we restrict to the special case of a homogeneous setting.
Else, we will write~$Y$ and~$\theta$, explicitly allowing locally varying parameter functions which satisfy the variability bound in Notation~\ref{not:inhomVsHom}.
Now we look for a description of the homogeneous propagation condition which enables an extension to the inhomogeneous setting. 
For this purpose, we introduce some auxiliary functions. 

\begin{Not}\label{not:funcP}
Let the functions $p^{(l)}_{\theta}: (0, \infty) \to [0,1]$, $l=1,2,3$ and $\theta \in \Theta$, be given as
\begin{eqnarray*}
	p^{(1)}_{\theta}(z) &:=& \mathbb{P}(\lbrace T(Y) > \theta \rbrace \cap \lbrace \mathcal{KL}(T(Y), \theta) > z \rbrace),\\
	p^{(2)}_{\theta}(z) &:=& \mathbb{P}(\lbrace T(Y) \leq \theta \rbrace \cap \lbrace \mathcal{KL}(T(Y), \theta) > z \rbrace), \quad z> 0\\
	p^{(3)}_{\theta}(z) &:=& \mathbb{P}(\lbrace T(Y) \leq \theta \rbrace \cap \lbrace \mathcal{KL}(T(Y), \theta) \leq z \rbrace),
\end{eqnarray*}
where $Y \sim \mathbb{P}_{\theta}$.
\end{Not}

\begin{Lem}\label{lem:funcP1}
The functions~$p^{(l)}_{\theta}$, $l=1,2,3$, in Notation~\ref{not:funcP} are invariant with respect to the parameter~$\theta \in \Theta$ for the Gaussian, log-normal, Gamma, Erlang, scaled chi-squared, exponential, Rayleigh, Weibull, and Pareto distributions.
\end{Lem}

The study in \cite[\S 4.1]{BecMat2013} suggests the invariance of the homogeneous propagation condition with respect to the parameter~$\theta$ for the Gaussian, log-normal, exponential, Rayleigh, Weibull, and Pareto distribution. 
In the following lemma, we take advantage of this invariance.
There, we completely determine the corresponding function~$\mathfrak{Z}_{\lambda}$ via the distribution of the positions of the observations around the respective parameter~$\vartheta \in \Theta$ given by the functions~$p^{(l)}_{\vartheta}$, $l=1,2,3$.

\begin{Lem}\label{lem:funcP2}
Assume the setting of Notation~\ref{not:inhomVsHom}. 
If the homogeneous propagation condition is invariant with respect to the parameter $\vartheta \in \Theta$, then the corresponding function~$\mathfrak{Z}_{\lambda}$ is uniquely determined by the functions~$p^{(l)}_{\vartheta}$, $l=1,2,3$, for every $\vartheta \in \Theta$.
\end{Lem}

In other words, the homogeneous propagation condition is determined by the probability distributions of~$\mathcal{KL}(T(\mathcal{Y}), \vartheta)$ on $\lbrace T(\mathcal{Y}) > \vartheta \rbrace$ and on $\lbrace T(\mathcal{Y}) \leq \vartheta \rbrace$. 
Under inhomogeneity, we  have to additionally compensate for the local variability of the parameter function. 
We investigate the interplay of the observations via the distribution of~$\mathcal{KL} (Y_i, Y_j)$,
which we compare with its homogeneous counterpart $\mathcal{KL} (\mathcal{Y}_i, \mathcal{Y}_j)$.
For simplicity, we presume the sufficient statistic~$T$ in Assumption~\ref{A1} to equal the identity. 
Instead of that, we could replace in the following all observations~$Y_i$ and~$\mathcal{Y}_i$ by the transformed observations~$T(Y_i)$ and~$T(\mathcal{Y}_i)$, leading, for every strictly monotonic~$T$, to the same results but more tedious terms.
We restrict to the favorable realizations, where the corresponding event~$M_0$ is related to the event~$\Omega_{\varkappa}$ in Corollary~\ref{cor:pKappa}.	

\begin{Prop}\label{prop:inhomPC}
Suppose the setting of Notation~\ref{not:inhomVsHom}, $T=\mathrm{Id}$, 
and let the functions~$p^{(l)}_{\theta}$, $l=1,2,3$, be invariant with respect to the parameter $\theta \in \Theta^*$. 
Additionally, recall Notation~\ref{not:varkappa}, and let $\varkappa \geq 1$ satisfy $\lbrace \vartheta \rbrace \cup \lbrace \theta_i \rbrace_{i=1}^n \in (\Theta_{\varkappa})^{n+1}$, where $\Theta_{\varkappa} \subseteq \Theta$ maximizes the probability of the event
\[
	M_0 := \bigcap_{i=1}^n\lbrace Y_i, \mathcal{Y}_i \in \Theta_{\varkappa} \rbrace. 
\]
Then, for all $z > \varkappa^2 \varphi_0^2$ and every $i,j \in \lbrace 1,...,n \rbrace$, it holds
\[
	\mathbb{P}\left( \lbrace \mathcal{KL} (Y_i, Y_j) > z \rbrace \vert M_0 \right) 
	\leq \mathbb{P}\left( \left\{ \varkappa^2 [ \varkappa \, \mathcal{KL}^{1/2} (\mathcal{Y}_i, \mathcal{Y}_j) + \varphi_0 ]^2 > z \right\} \vert M_0\right).
\]
\end{Prop}

Now we propose a precise choice of the adaptation bandwidth for the case of a (piecewise) bounded parameter function. 

\begin{Claim}\label{claim}
Let Assumption~\ref{A1} be satisfied with $T: \mathcal{Y} \to \mathbb{R}$ strictly monotonic, and fix a subset $\Theta^* \subseteq \Theta$ and some constant $\varphi_0 := \varphi / \max_i (\overline{N}_i^{(k^*)} )^{1/2}$ with $\varphi > 0$. 
Additionally, let the homogeneous propagation condition and the functions~$p^{(l)}_{\theta}$, $l=1,2,3$, be invariant with respect to the parameter $\theta \in \Theta$. 
Finally, we presume the adaptation bandwidth $\lambda>0$ to be in accordance with the homogeneous propagation condition at level $\epsilon > 0$.
Then, the choice
\[
	\lambda_{\varphi} := \varkappa^4 \left[ \sqrt{\lambda} + \varphi \right]^2
\]
is in accordance with the inhomogeneous propagation condition at probability level $\epsilon (\lambda_{\varphi}) \leq \epsilon + 2 p_{\varkappa}$ and variability level~$\varphi_0$ for the parameter set~$\Theta^*$, where~$p_{\varkappa}$ is as in Notation~\ref{not:pKappa}.
\end{Claim}

Admittedly, the iterative approach of the algorithm impedes a definite proof. 
Instead, we will present a justification of Claim~\ref{claim} in Appendix~\ref{app:proofs}, where we will follow an inductive argumentation in order to overcome the remaining gap at least to a certain extent.

\begin{Rem}
Recall that the parameter set $\Theta^* \subseteq \Theta$ and the variability level~$\varphi_0$ influence the sizes of the constant $\varkappa \geq 1$ and the corresponding probability~$p_{\varkappa}$.
Moreover, we point out that the probability level $\epsilon + 2 p_{\varkappa}$ is an upper bound for the actual probability level~$\epsilon (\lambda_{\varphi})$, but this bound does not have to be sharp.
Similarly, the proposed choice of~$\lambda_{\varphi}$ is based on a rough estimation, where the effectively required values of~$\varkappa$ and~$\varphi$ may be much smaller than supposed.
Hence, in practice, one can always use a smaller bandwidth $\lambda^* < \lambda_{\varphi}$ with an unknown probability level if this seems to be advantageous.
This may increase the risk of adaptation to noise as usually  $\epsilon(\lambda^*) \geq \epsilon (\lambda_{\varphi})$, but the main property, propagation with probability $1-\epsilon(\lambda^*)$, remains valid. 
In any case, one should use the homogeneous adaptation bandwidth~$\lambda$ as a lower bound, $\lambda \leq \lambda^*$.
Recalling Example~\ref{ex:pKappa} concerning the trade-off between~$\varkappa$ and~$p_{\varkappa}$, we conclude the following.
\begin{compactitem}
\item Claim~\ref{claim} provides a reasonable choice of the adaptation bandwidth 
if $\varkappa = 1$ and $p_{\varkappa} = 0$, such as for Gaussian and log-normal distributed observations.
\item For Gamma, Erlang, scaled chi-squared, exponential, Rayleigh, Weibull, and Pareto distribution, $\varkappa$~and $p_{\varkappa}$ are large.
For these distributions, Claim~\ref{claim} justifies the \emph{existence} of an adaptation bandwidth~$\lambda_{\varphi}$ which is in accordance with the inhomogeneous propagation condition at level~$ \epsilon(\lambda_{\varphi}) \leq \epsilon + 2 p_{\varkappa}$, but its practical use is questionable due to the sizes of~$\varkappa$ and~$p_{\varkappa}$. 
\end{compactitem}
\end{Rem}

\section{Simulations}\label{sec:simulations}

In \citep{BecMat2013} and Section~\ref{sec:Theory}, we established several theoretical properties of the simplified Propagation-Separation Approach in Notation~\ref{algorithm}.
Here, we will illustrate these properties by simulated examples with Gaussian and exponentially distributed observations.
In particular, we will compare the results of the simplified algorithm with the original procedure in Algorithm~\ref{algorithmMS}.
We will present several example plots for the realization $\texttt{seed=1}$, the corresponding weighting schemes, and boxplots of the mean absolute error (MAE) over $1000$ realizations, where $\texttt{seed=l}$ and $l \in \lbrace 1,..., 1000 \rbrace$.

\subsection{Test functions and methods}\label{sec:numerics-testf}

Here, we present all test functions that we will consider in the following numerical study.
Usually, we simulated data with $n = 1000$ observations.
In some examples, we changed the sample size $n \in \mathbb{N}$.
Then, we increased the cardinality of each region of the introduced parameter functions by the same factor such that the design portions remained unchanged. 

We used the implementation of the Propagation-Separation Approach in the R-package \pkg{aws} by \citet{aws}.
Here, the memory step is omitted by default.
If desired, it can be included in the procedure setting \texttt{memory=TRUE}.
The memory step is implemented for two different memory kernels, which can be specified by \texttt{aggkern=\textquotedbl{}Triangle\textquotedbl{}} or \texttt{aggkern=\textquotedbl{}Uniform\textquotedbl{}}.
If not mentioned differently, we applied the default parameters using the command
\begin{eqnarray*}
\texttt{that} & \texttt{<-} & \texttt{aws(tnoise, u=t, hmax=10000, } \\ 
&& \texttt{homogen=FALSE, maxni=TRUE)},
\end{eqnarray*}
where \texttt{tnoise} denotes the simulated observations and \texttt{t} is the corresponding expectation.
For the sake of simplicity, we show only univariate examples where $\mathcal{X} \subseteq \mathbb{R}$.

By means of an additionally included function \texttt{awsweights}, we visualized the weighting schemes of the resulting non-adaptive weights~$\lbrace \overline{w}_{ij}^{(k)} \rbrace_{i,j}$, the adaptive weights~$\lbrace \tilde{w}_{ij}^{(k)} \rbrace_{i,j}$, 
and the adaptation kernel $\lbrace K_{\mathrm{ad}} (s_{ij}^{(k)} / \lambda) \rbrace_{i,j}$, which equals the quotient $\tilde{w}_{ij}^{(k)} / \overline{w}_{ij}^{(k)} = K_{\mathrm{ad}} (s_{ij}^{(k)} / \lambda)$ if $\overline{w}_{ij}^{(k)} > 0$.
The values of these quantities are shown in grey scales, where zero corresponds to black and one to white, respectively.
Moreover, we included a scaling factor \texttt{tadjust}, which allows a manipulation of the memory bandwidth~$\tau$.
In the package \pkg{aws}, this is given as 
\[
	\tau^{(k)} := (2*\tau_1+\tau_1 * \max \lbrace k_{\text{star}} - \log (h^{(k)}), 0 \rbrace), 
\]
where the constant~$k_{\text{star}}$ depends on the family of probability distributions~$\mathcal{P}$, and
\[
	\tau_1 := \begin{cases}
	\texttt{ladjust*tadjust*20}, \quad & \text{ if } \texttt{aggkern=\textquotedbl{}Triangle\textquotedbl{}} \\
	\texttt{ladjust*tadjust*8}, \quad & \text{ if } \texttt{aggkern=\textquotedbl{}Uniform\textquotedbl{}}
	\end{cases}
\]
with $\texttt{ladjust=1}$ and $\texttt{tadjust=1}$ by default.

We applied three test functions where the structural Assumption~\ref{ASinhom} is violated.
First, we used a piecewise smooth function, given as
\begin{equation}\label{eq:polynTestf}
	\theta (x) := \begin{cases} 
		7 + x/250 \quad &\text{if } x \in \lbrace 1,..., 250 \rbrace, \\
		11 + ((x-450)/100)^2/2 \quad &\text{if } x \in \lbrace 251,..., 750 \rbrace, \\
		6 - (x-750)/200 \quad &\text{if } x \in \lbrace 751,..., 1000 \rbrace. 
	        \end{cases}
\end{equation}
Second, we used a piecewise constant function with small discontinuities and three different regions of monotonicity,
\begin{equation}\label{eq:updownTestf}
\begin{array}{lllll}
	\theta(x) := 0, & x \in \lbrace 1,..., 50 \rbrace, & \qquad &
	\theta(x) := 2.2, & x \in \lbrace 451,..., 500 \rbrace, \\
	\theta(x) := 0.5, & x \in \lbrace 51,..., 100 \rbrace, & \qquad &
	\theta(x) := 1.7, & x \in \lbrace 501,..., 550 \rbrace, \\
	\theta(x) := 1, & x \in \lbrace 101,..., 150 \rbrace, & \qquad &
	\theta(x) := 1.2, & x \in \lbrace 551,..., 600 \rbrace, \\
	\theta(x) := 1.5, & x \in \lbrace 151,..., 200 \rbrace, & \qquad &
	\theta(x) := 0.7, & x \in \lbrace 601,..., 650 \rbrace, \\
	\theta(x) := 2, & x \in \lbrace 201,..., 250 \rbrace, & \qquad &
	\theta(x) := 0.9, & x \in \lbrace 651,..., 750 \rbrace, \\
	\theta(x) := 2.5, & x \in \lbrace 251,..., 300 \rbrace, & \qquad &
	\theta(x) := 1.6, & x \in \lbrace 751,..., 800 \rbrace, \\
	\theta(x) := 3, & x \in \lbrace 301,..., 350 \rbrace, & \qquad &
	\theta(x) := 2.6, & x \in \lbrace 801,..., 900 \rbrace, \\
	\theta(x) := 3.2, & x \in \lbrace 351,..., 400 \rbrace, & \qquad &
	\theta(x) := 2.9, & x \in \lbrace 901,..., 1000 \rbrace. \\
	\theta(x) := 2.7, & x \in \lbrace 401,..., 450 \rbrace, & & &
\end{array}
\end{equation}
This function is constructed especially to illustrate the consequences of close steps in distant locations.
Third, we will study the behavior of the simplified Propagation-Separation Approach for the logarithmic function
\begin{equation}\label{eq:logTestf}
	\theta(x) := \log(x), \qquad x \in \mathcal{X}.
\end{equation}
Here, the parameter values change slowly.

Additionally, we consider a shifted and scaled indicator function.
This piecewise constant setting coincides with the setting that the original Propagation-Separation Approach in \citep{PoSp05} and its simplified version in Notation~\ref{algorithm} assume. 
Let the sample size be even, $n \in 2 \mathbb{N}$.
Then, we split the design into two parts with coinciding cardinality, $\mathcal{X}_1 := \lbrace X_i \rbrace_{i=1}^{n/2}$ and $\mathcal{X}_2 := \lbrace X_i \rbrace_{i=n/2 +1}^{n}$.
We consider the test function 
\begin{equation}\label{eq:indiTestf}
	\theta(x) := 1 + 4 \cdot \boldsymbol{1}_{\mathcal{X}_2}(x), \qquad x \in \mathcal{X},
\end{equation}
where~$\boldsymbol{1}$ denotes the indicator function.

%First, we illustrate the formation of the associated step function in Definition~\ref{def:assStepf}.
%In \S~\ref{sec:numerics-memory}, we concentrate on the impact of the memory step.
%For this purpose, we compare the results of the original and the simplified Propagation-Separation Approach on the following test functions.
%In particular, we vary the memory bandwidth in order to illustrate the effects of an increasing amount of aggregation.

\subsection{Formation of the associated step function}\label{sec:simStep}

We study the formation of the associated step function, which we introduced in \S~\ref{sec:stability}.
For this purpose, we visualize the resulting weighting schemes.

In Figure~\ref{fig:weightsPolynG}, we consider the piecewise smooth function~\eqref{eq:polynTestf} with Gaussian observations.
In the first row, we provide the weighting schemes of the iteration step where the MAE is minimized.
The product of the adaptive term $\lbrace K_{\mathrm{ad}} (s_{ij}^{(k)} / \lambda) \rbrace_{i,j}$~(a) and the non-adaptive weights~$\lbrace \overline{w}_{ij}^{(k)} \rbrace_{i,j}$~(b) results in the adaptive weights $\lbrace \tilde{w}_{ij}^{(k)} \rbrace_{i,j}$~(c).
This illustrates the interaction of adaptation and location.
For $\texttt{hmax=2000}$, the algorithm results in the associated step function~(d).
Here, the adaptive weights~(f) and the weighting scheme of the corresponding adaptive term $\lbrace K_{\mathrm{ad}} (s_{ij}^{(k)} / \lambda) \rbrace_{i,j}$ (not shown) were visually indistinguishable due to the large size of the considered local neighborhood, which is determined by the non-adaptive weighting scheme~(e).

We present in Figure~\ref{fig:weightsUpdownG} the example plots and corresponding weighting schemes~$\lbrace \tilde{w}_{ij}^{(k)} \rbrace_{i,j}$ of the step function~\eqref{eq:updownTestf}.
The discontinuities are too small for separation.
Therefore, the algorithm forms a step function which differs from the original parameter function.
The minimal MAE is provided for $\texttt{hmax=30}$, where the considered local neighborhood is small and separation does not yet occur~(a+e).
We observe, in the example plot~(b) as well as in the weighting scheme~(f), that the estimation function starts to form a step function for $\texttt{hmax=120}$.
In~(c), the estimation function of $\texttt{hmax=600}$ resembles a step function, but the weighting scheme~(g) already indicates that the formed steps may change with increasing location bandwidths.
Indeed, in~(d), several steps in different locations have been assimilated as the weighting scheme~(h) points out.
For the plots in the last column, we set $\texttt{hmax=20000}$.

\begin{figure}
\begin{center}
\includegraphics[width = 0.72 \textwidth]{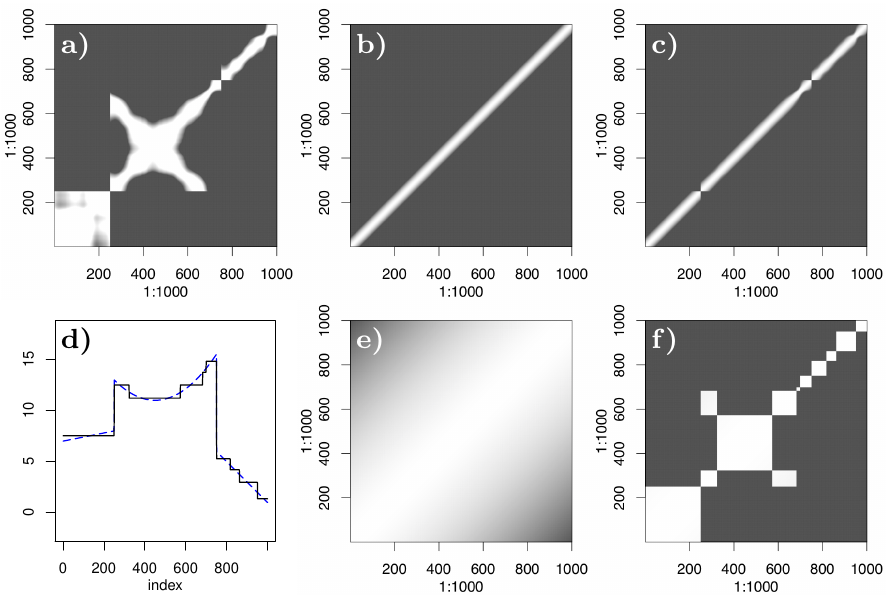} 
\caption[Formation of the associated step function: Gaussian observations, example~1]{Formation of the associated step function for the piecewise smooth function~\eqref{eq:polynTestf} with Gaussian observations.}
\label{fig:weightsPolynG}
\end{center}
\end{figure}

\begin{figure}
\begin{center}
\includegraphics[width = 0.95 \textwidth]{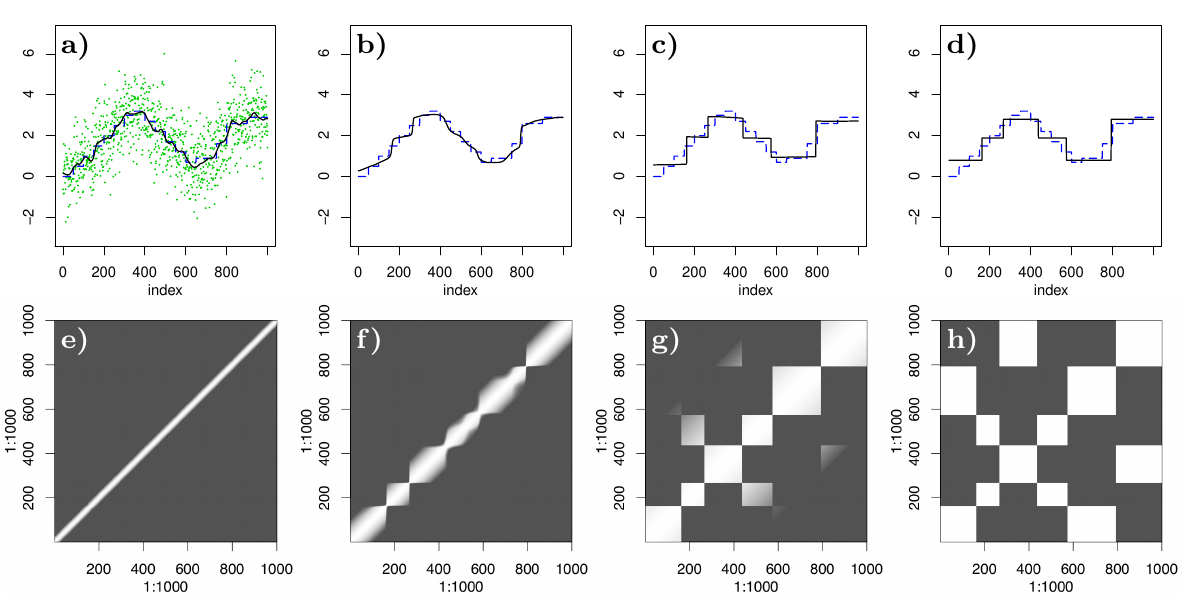} 
\caption[Formation of the associated step function: Gaussian observations, example~2]{Formation of the associated step function for the step function~\eqref{eq:updownTestf} with Gaussian observations.}
\label{fig:weightsUpdownG}
\end{center}
\end{figure}

\begin{figure}
\begin{center}
\includegraphics[width = 0.95 \textwidth]{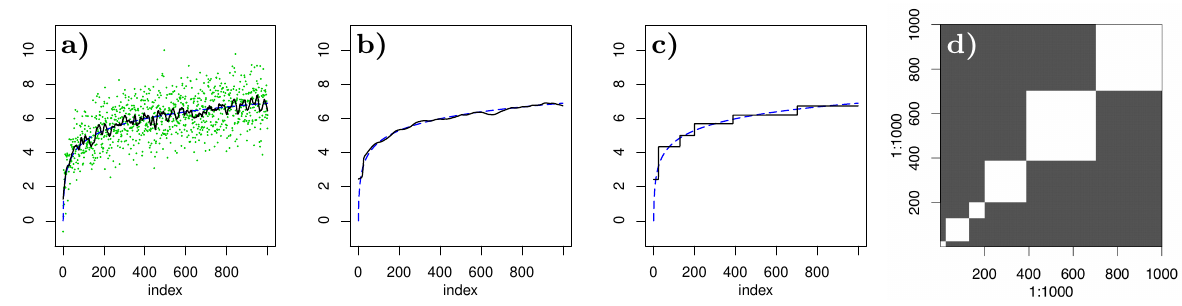} 
\caption[Formation of the associated step function: Gaussian observations, example~3]{Formation of the associated step function for the logarithmic function~\eqref{eq:logTestf} with Gaussian observations.}
\label{fig:weightsLogG}
\end{center}
\end{figure}

Even for the logarithmic function~\eqref{eq:logTestf}, the simplified algorithm results in a step function with disjoint regions~(d).
In Figure~\ref{fig:weightsLogG}, we show the example plots for a small location bandwidth $\texttt{hmax=10}$~(a), at $\texttt{hmax=60}$~(b), where the MAE is minimal, and at $\texttt{hmax=2000}$~(c).

\begin{figure}
\begin{center}
\includegraphics[width = 0.95 \textwidth]{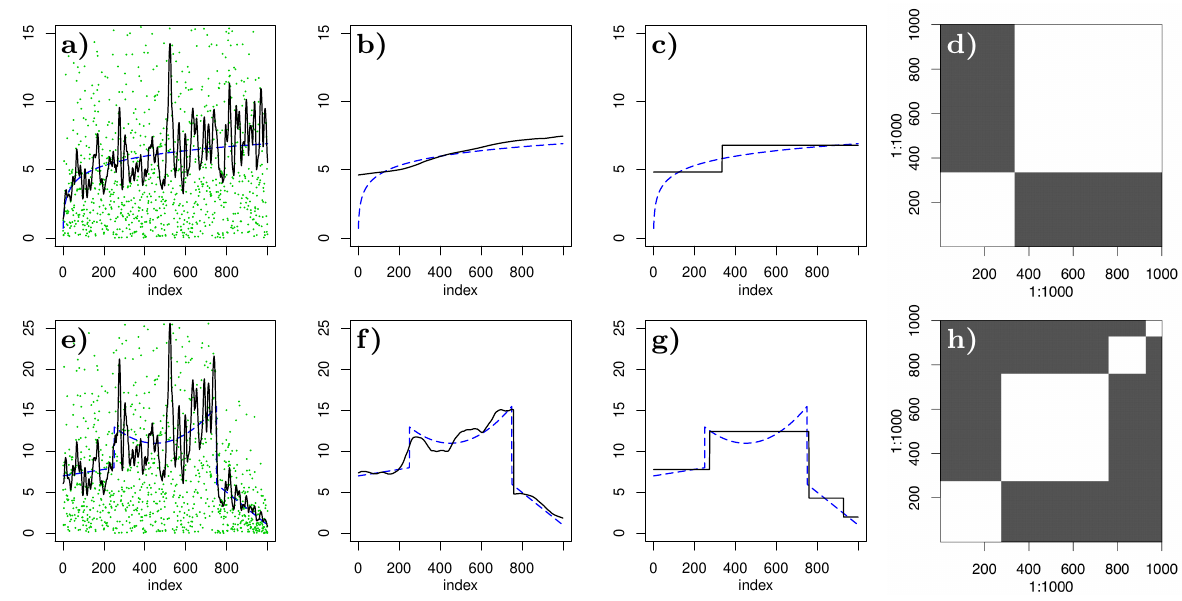} 
\caption[Formation of the associated step function: Exponentially distributed observations]{Formation of the associated step function for the piecewise smooth function~\eqref{eq:polynTestf} and the logarithmic function~\eqref{eq:logTestf} with exponentially distributed observations.}
\label{fig:weightsExp}
\end{center}
\end{figure}

Additionally, we studied the formation of the associated step function for exponentially distributed observations on several test functions. 
In Figure~\ref{fig:weightsExp}, we provide the results for the parameter functions in Equations~\eqref{eq:polynTestf} (first row) and~\eqref{eq:logTestf} (second row).
Here again, for sufficiently large location bandwidths, the algorithm results in the associated step function with disjoint regions of non-zero adaptive weights~(d+h).
We show the example plots for a small location bandwidth (a+e), where $\texttt{hmax=10}$, an intermediate iteration step with minimized MAE (b+f), and a large location bandwidth $\texttt{hmax=20000}$ (c+g).

\subsection{Impact of the memory step}\label{sec:numerics-memory}

In \citep[Thm. 5.7]{PoSp05}, the memory step provided a general result on the stability of estimates, up to some constant.
However, its practical use is questionable.
No situation has been reported to date where the memory step considerably improved the results of the Propagation-Separation Approach.
Therefore, we aim for a better understanding of its impact on the resulting estimates.
For this purpose, we compared the results of the original and the simplified algorithm on the test functions in \S~\ref{sec:numerics-testf} for Gaussian and exponentially distributed observations.

In Figure~\ref{fig:GaussMemHBox}, we show the results for the piecewise smooth function~\eqref{eq:polynTestf} with Gaussian distributed observations.
Here, we applied three location bandwidths, $\texttt{hmax=50,500,5000}$, each of them without memory step (\texttt{memory=FALSE}) and with memory step (\texttt{memory=TRUE}), using a triangular kernel (\texttt{aggkern=\textquotedbl{}Triangle\textquotedbl{}}) and a uniform kernel (\texttt{aggkern=\textquotedbl{}Uniform\textquotedbl{}}).
As for all other test functions with Gaussian or exponentially distributed observations, there is (almost) no difference between the resulting boxplots with and without memory step and for the two memory kernels.
This raises the question whether the memory step itself does not have any effect, or whether the default parameter choices in the \textbf{R}-package \pkg{aws} are unfavorable.
\begin{figure}
\begin{center}
\includegraphics[width = 0.7 \textwidth]{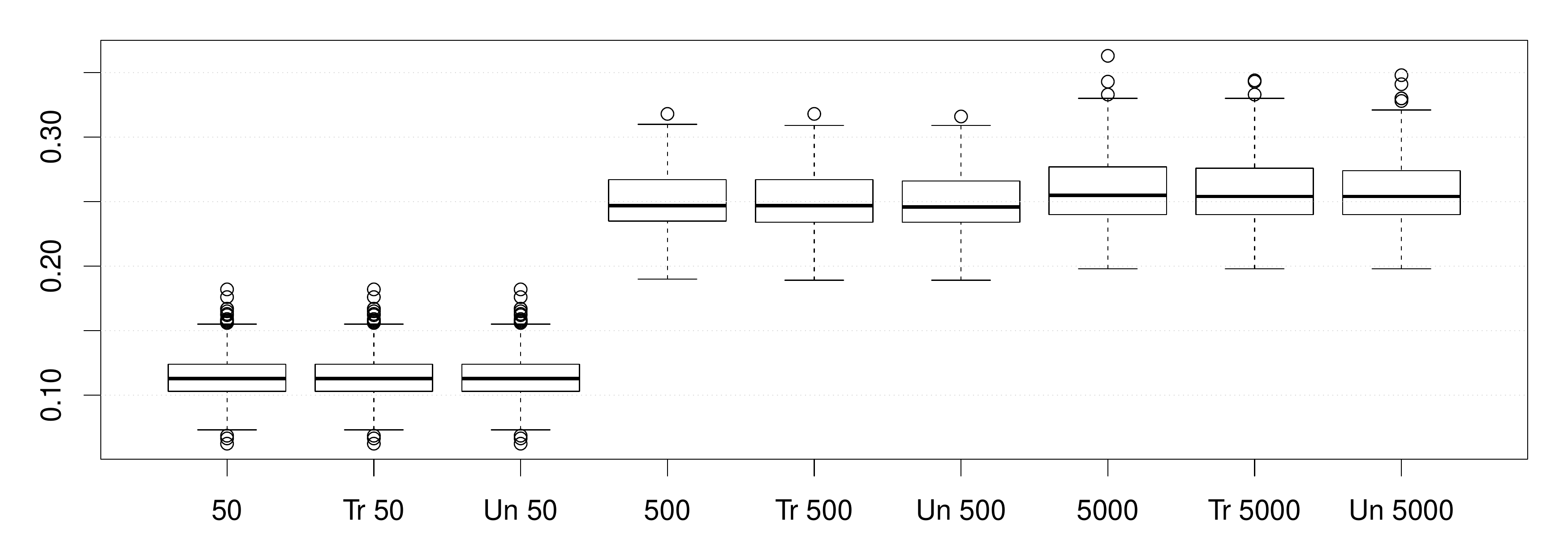}
\caption[Boxplots concerning the impact of the memory step: Default parameter choices]{MAE-boxplots for $\texttt{hmax=50,500,5000}$ with and without memory step, setting \texttt{aggkern=} \texttt{\textquotedbl{}Triangle\textquotedbl{}} (Tr) and \texttt{aggkern=\textquotedbl{}Uniform\textquotedbl{}} (Un).}
\label{fig:GaussMemHBox}
\end{center}
\end{figure}

\begin{figure}
\begin{center}
\includegraphics[width = 0.49 \textwidth]{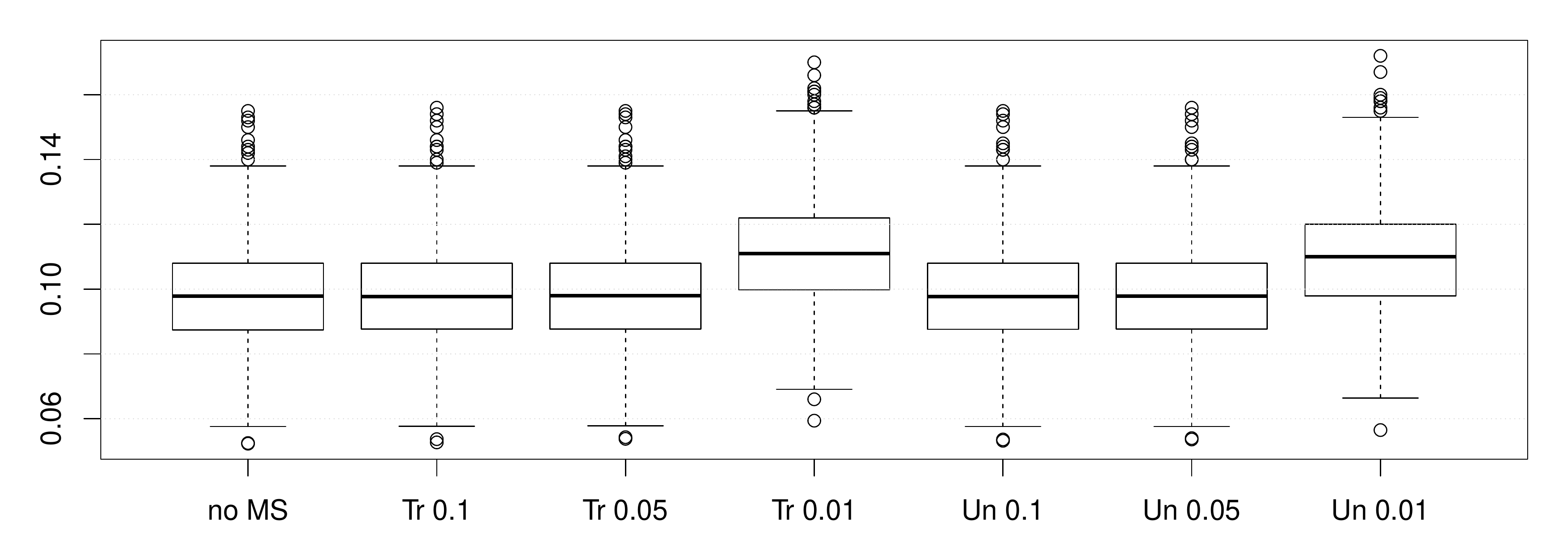}
\includegraphics[width = 0.49 \textwidth]{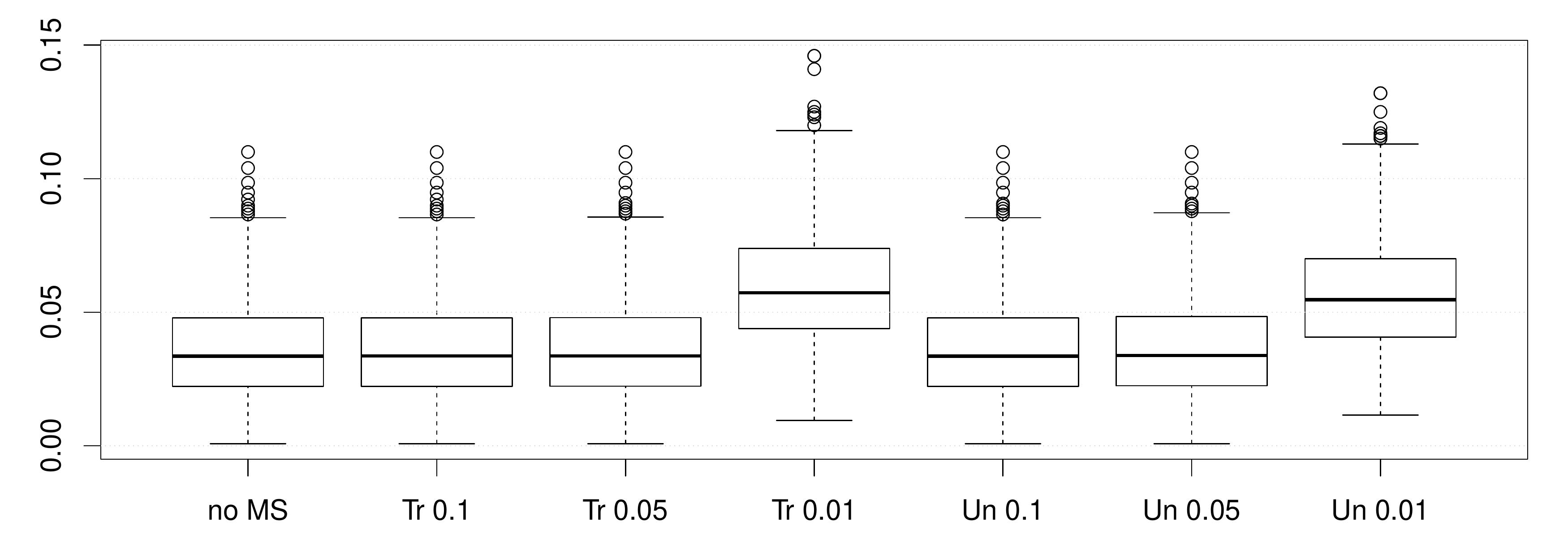} \\ 
\includegraphics[width = 0.49 \textwidth]{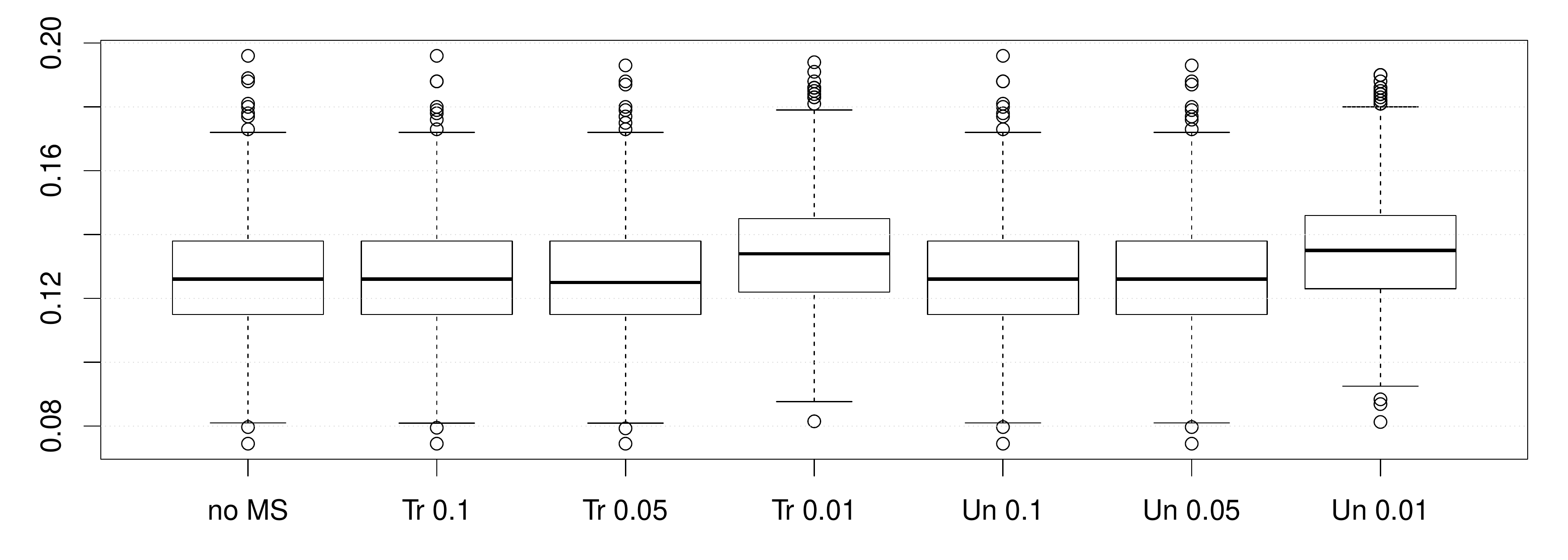}
\includegraphics[width = 0.49 \textwidth]{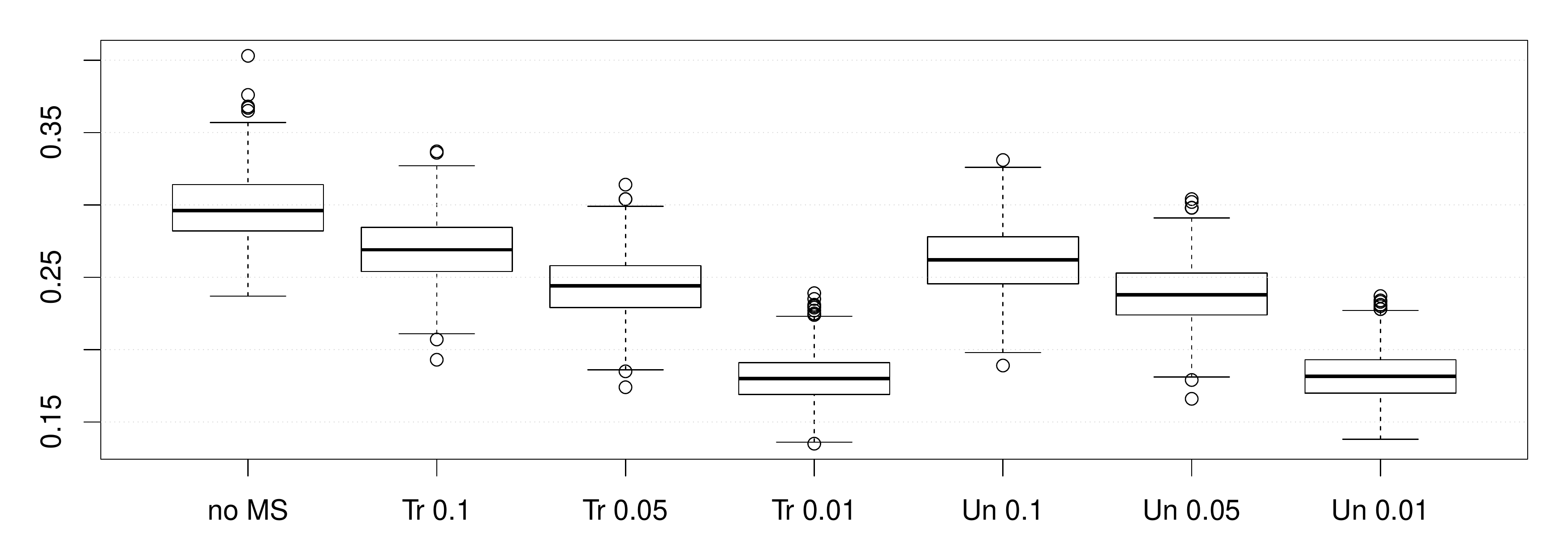}
\caption[Boxplots concerning the impact of the memory step: Increasing amount of aggregation]{MAE-boxplots for $\texttt{hmax=50}$ (left) and $\texttt{hmax=10000}$ (right) for the indicator function (top) and on the piecewise smooth function (bottom) with Gaussian observations.
We applied the algorithm without (no MS) and with memory step, setting \texttt{aggkern=\textquotedbl{}Triangle\textquotedbl{}} (Tr) and \texttt{aggkern=\textquotedbl{}Uniform\textquotedbl{}} (Un), where $\texttt{tadjust=0.1,0.05,0.01}$ increases the amount of aggregation.}
\label{fig:GaussMemBox}
\end{center}
\end{figure}

In order to provide a deeper insight into the mode of action of the memory step, we increased the amount of aggregation by means of the additionally implemented scaling factor \texttt{tadjust} of the memory bandwidth $\tau > 0$.
In Figure~\ref{fig:GaussMemBox}, we present, for the indicator function~\eqref{eq:indiTestf} (top) 
and for the piecewise smooth function~\eqref{eq:polynTestf} (bottom), the MAE-boxplots at some early iteration step with $\texttt{hmax=50}$ (left) and for $\texttt{hmax=10000}$ (right), assuming Gaussian observations. 
The amount of aggregation increases due to the choices $\texttt{tadjust=0.1,0.05,0.01}$ for \texttt{memory=TRUE} with \texttt{aggkern=\textquotedbl{}Triangle\textquotedbl{}} and \texttt{aggkern=\textquotedbl{}Uniform\textquotedbl{}}.
For comparison, we show the result of the simplified procedure as well, where \texttt{memory=FALSE}.
For $\texttt{hmax=50}$, we observe an increase of the MAE for both test functions at $\texttt{tadjust=0.01}$, while the MAE without memory step coincides with the results for $\texttt{tadjust=0.1}$ and $\texttt{tadjust=0.05}$.
For $\texttt{hmax=10000}$, this observation remains valid for the indicator function (top right), where the locally constant model of the Propagation-Separation Approach is satisfied.

In contrast, for the piecewise smooth function, we know from \S~\ref{sec:simStep} that the estimation function approaches the associated step function, which leads to an increase of the MAE. 
As demonstrated in the bottom right of Figure~\ref{fig:GaussMemBox}, the MAE decreases with increasing amount of aggregation, that is with decreasing $\texttt{tadjust}$.
Unfortunately, this increases the risk of adaptation to noise as we illustrate on some example plots in Figure~\ref{fig:MSGauss} for the piecewise smooth function (top) and for the step function~\eqref{eq:updownTestf} (bottom), both with Gaussian observations.
Without memory step as well as with memory step and $\texttt{tadjust=1}$, the algorithm results in the associated step function (a+e).
To some extent, this effect can be attenuated by increasing the amount of aggregation, setting $\texttt{tadjust=0.05}$ (b+f) or even $\texttt{tadjust=0.02}$ (c+g).
For $\texttt{tadjust=0.01}$ (d+h), we observe adaptation to noise, which indicates the increased risk of adaptation to outliers due to the decreased memory bandwidth.
Naturally, for other realizations, larger sample sizes or different test functions, this could happen for larger values of \texttt{tadjust} as well.
We got similar results for the other test functions in \S~\ref{sec:numerics-testf} with Gaussian and as well with exponentially distributed observations (not shown).
\begin{figure}
\begin{center}
\includegraphics[width = 0.95 \textwidth]{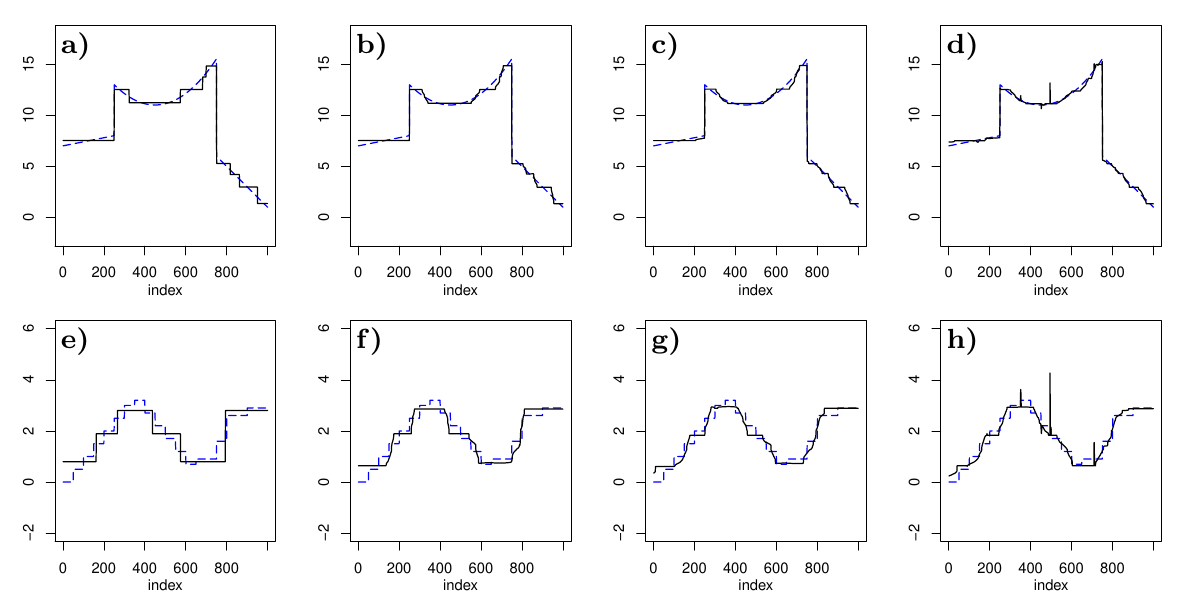} 
\caption[Example plots concerning the impact of the memory step]{Example plots for the piecewise smooth function (top) and the step function (bottom) at $\texttt{hmax=100000}$. 
We applied the algorithm without (a+e) and with memory step, setting \texttt{aggkern=\textquotedbl{}Triangle\textquotedbl{}} and (from left to right) \texttt{tadjust= 0.05,0.02,0.01}.}
\label{fig:MSGauss}
\end{center}
\end{figure}

\subsection{Stability of estimates}\label{sec:numerics-stability}

The numerical results in \S~\ref{sec:simStep} suggest that the simplified Propagation Separation Approach provides a certain stability of estimates, where the associated step function acts as an intrinsic stopping criterion.
In \S~\ref{sec:convergence}, we will discuss the reasons which impede a theoretical proof of this heuristic property.
Here, we present some boxplots which indicate the immutability of the MAE for sufficiently large location bandwidths.
 
We show results for Gaussian and exponentially distributed observations (Figure~\ref{fig:stabilityGaussExp}). 
We set $\texttt{n=1000}$ and $\texttt{hmax=20,50,200,500,1000,10000,20000}$ for the former, and $\texttt{n=4000}$ and \texttt{hmax= 50,200,500,2000,15000,20000} for the latter.
Here again, we consider the indicator function~\eqref{eq:indiTestf}, where the structural assumption of the Propagation-Separation Approach is satisfied. 
This leads to a decreasing MAE during iteration.
As an example for the case of a misspecified model, we again apply the piecewise smooth function~\eqref{eq:polynTestf}.
Here, the MAE increases for larger location bandwidth as the estimator is forced into a step function.
Nevertheless, for both test functions and both probability distributions, the MAE stabilizes for sufficiently large location bandwidths.
For comparison, we show the MAE which results from the choices \texttt{memory=TRUE,} \texttt{aggkern=\textquotedbl{}Triangle\textquotedbl{},} \texttt{tadjust=1}.

\begin{figure}
\begin{center}
\includegraphics[width = 0.49 \textwidth]{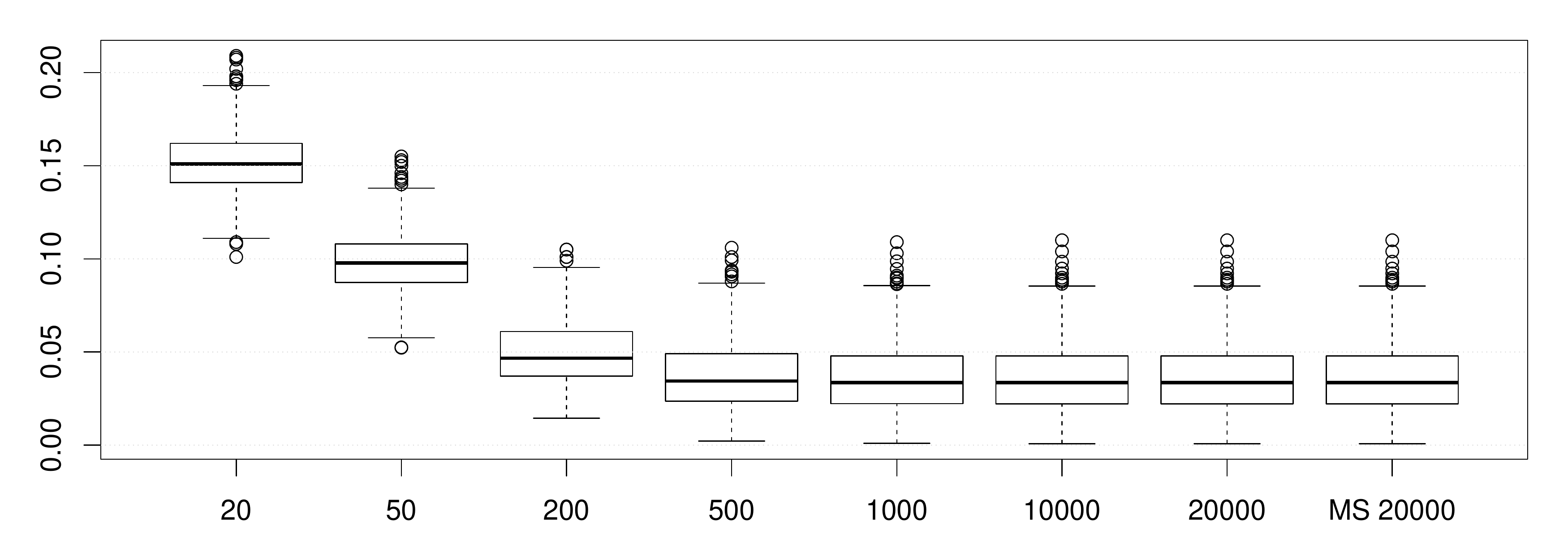}
\includegraphics[width = 0.49 \textwidth]{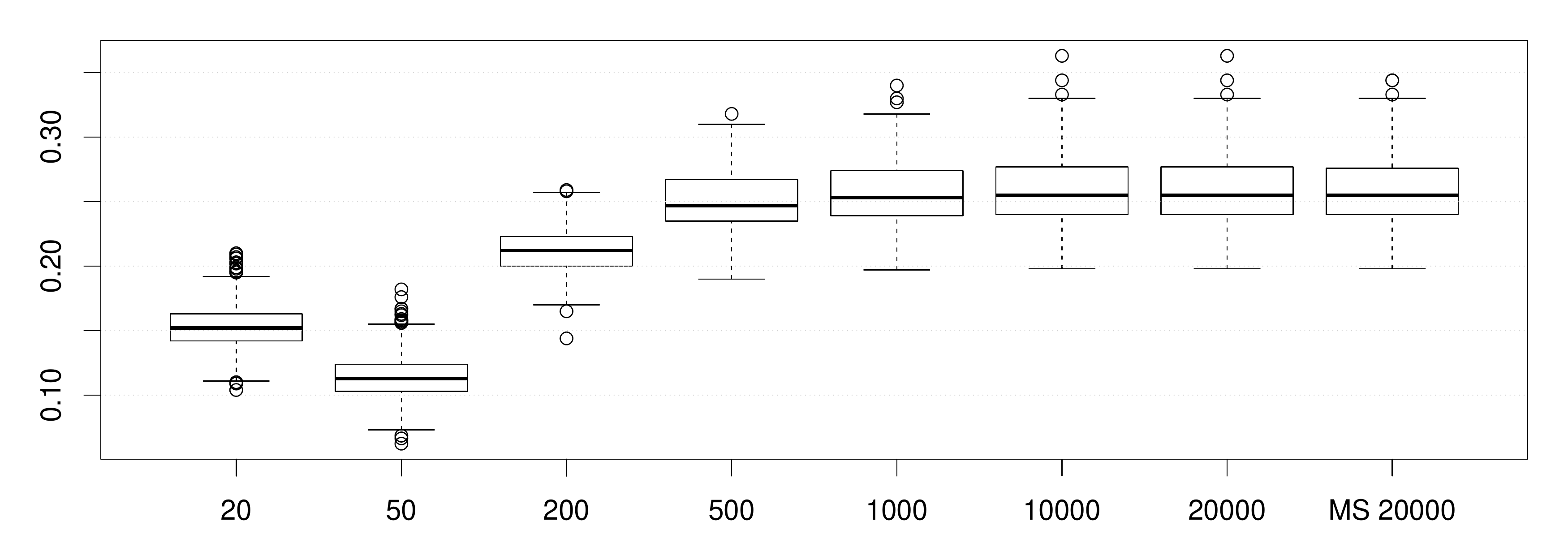} \\
\includegraphics[width = 0.49 \textwidth]{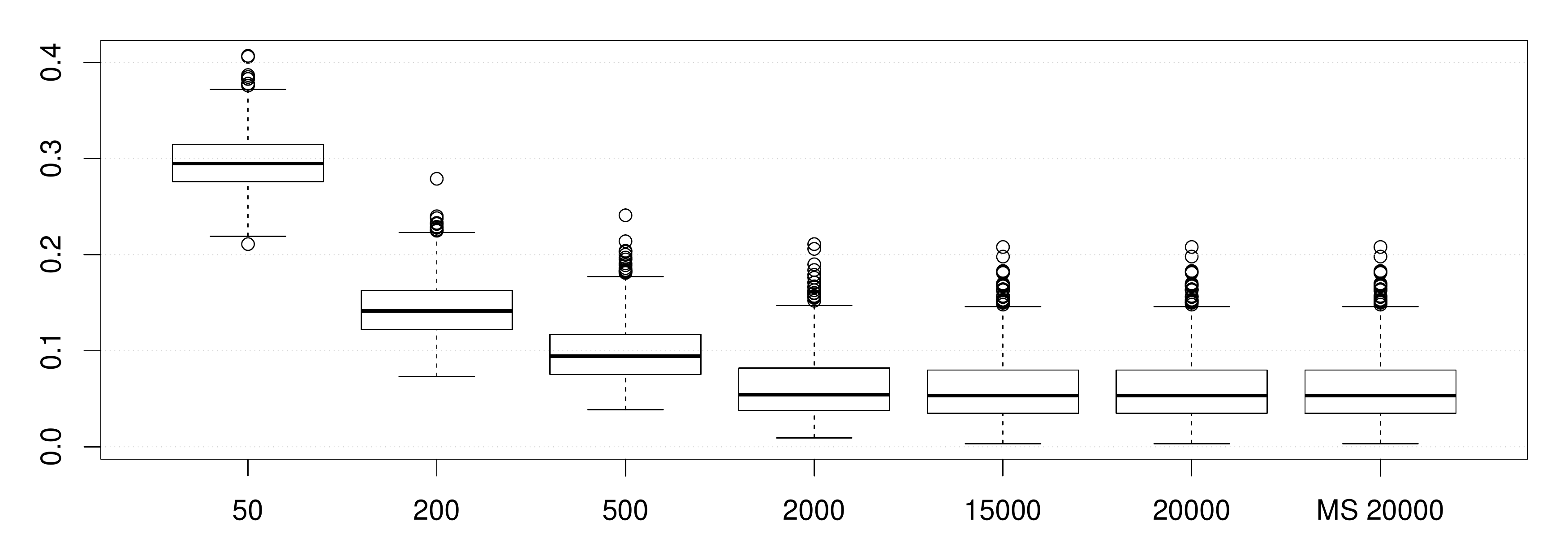}
\includegraphics[width = 0.49 \textwidth]{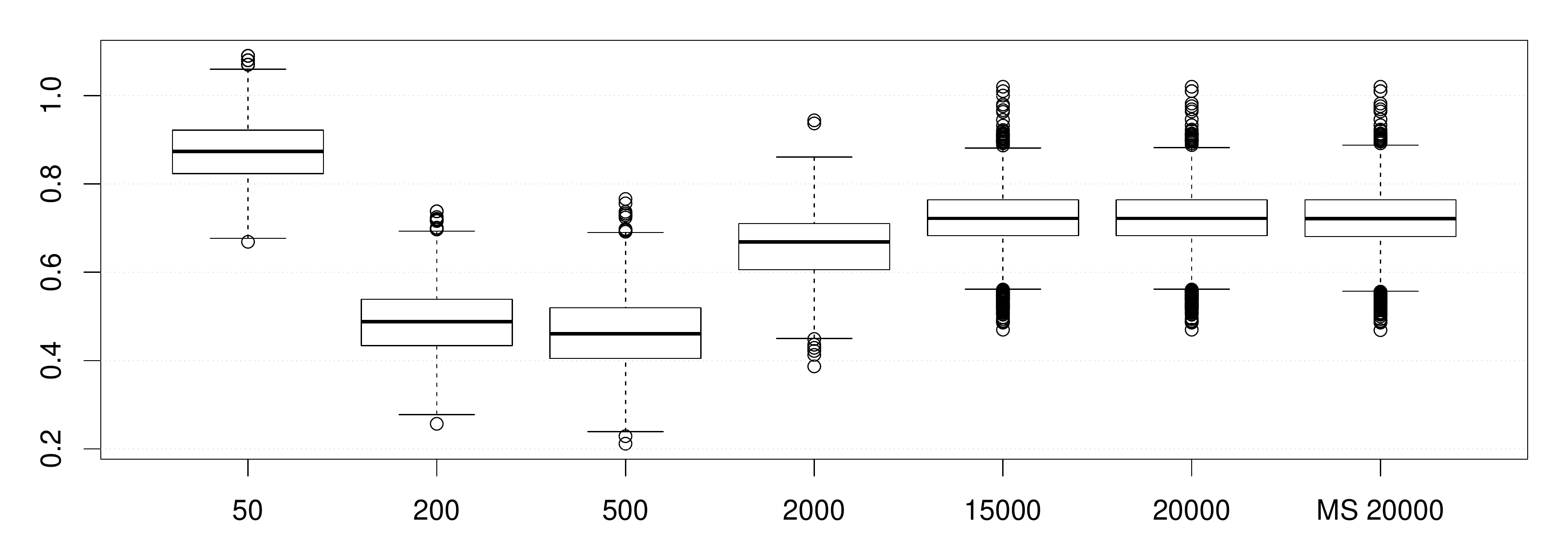}
\caption[Boxplots concerning the stability of estimates]{Stability of estimates for Gaussian (top) and exponentially (bottom) distributed observations for the indicator function (left) and for the piecewise smooth function (right).}
\label{fig:stabilityGaussExp}
\end{center}
\end{figure}

\section{Discussion of the simplified Propagation-Separation Approach}\label{sec:PS-discussion}

Our study provides theoretical and numerical results for the simplified Pro\-pa\-ga\-tion-Sepa\-ra\-tion Approach in Notation~\ref{algorithm}, where the memory step is omitted. 
This helps for a better understanding of the procedure as the impact and interaction of the involved components is clarified.
Furthermore, the presented results substantiate the reasons for omitting the memory step and provide, for the first time, a detailed study of its impact.

Next we will discuss the following two questions.
\begin{compactitem}
\item Does the simplified Propagation-Separation Approach converge?
\item (Where) do we need the memory step?
\end{compactitem}
Finally, we will give a brief overview on possible topics for future research.

\subsection{Does the Propagation-Separation Approach converge?}\label{sec:convergence}

In \S~\ref{sec:stability}, we introduced a specific step function, which approximates the estimation function from the Propagation-Separation Approach. 
The formation of this \emph{associated step function} can be explained as follows.

Due to the support $\mathrm{supp} (K_{ad}) = [0,1)$, the statistical penalty~$s_{ij}^{(k)}$ defined in Algorithm~\ref{algorithmMS} (page~\pageref{algorithmMS}) ensures zero weights~$\tilde{w}_{ij}^{(k)} = 0$ 
if the Kullback-Leibler divergence of the estimators from the last iteration step $\mathcal{KL} ( \tilde{\theta}_i^{(k-1)}, \tilde{\theta}_j^{(k-1)} )$ exceeds some lower bound~$\lambda / \tilde{N}_i^{(k-1)}$.
Let us consider the case where $\overline{w}_{ij}^{(k^*)} > 0$ implies $\tilde{w}_{ij}^{(k^*)} > 0$ for all $i,j \in \lbrace 1,...,n \rbrace$. 
This means that separation did not occur. 
The monotonicity of the sequence of location bandwidths~$\lbrace h^{(k)} \rbrace_{k=0}^{k^*}$ ensures that the non-adaptive weights increase during iteration, and, without separation, all estimators approach each other. 
With~$h^{(k^*)}$ sufficiently large, this results in an almost constant estimation function.
However, in many cases, there are $X_i, X_j \in \mathcal{X}$ such that $\overline{w}_{ij}^{(k^*)} > 0$, but $\tilde{w}_{ij}^{(k^*)} = 0$. 

We know from Proposition~\ref{prop:separationMS} that separation occurs if $\max \lbrace \mathcal{KL} \left( \theta_i, \theta_j \right): X_i, X_j \in \mathcal{X} \rbrace$ is sufficiently large, or if the algorithm adapts to outliers. 
The latter leads to separation of single observations, probably together with some local neighborhood. 
If separation happens due to the variability of the true parameter function~$\theta(.)$, then it starts in regions with a large absolute value of the first derivative~$\vert \theta'(.) \vert$.
Beginning either at the boundaries of the design space~$\mathcal{X}$ or close to discontinuities and local extrema of~$\theta(.)$ and~$\vert \theta'(.) \vert$, this leads, by subsequent attraction and repulsion of the estimators, to the formation of a step function which approximates the associated step function in Definition~\ref{def:assStepf}.

On the test functions in \S~\ref{sec:numerics-testf}, we observed for Gaussian and exponential distributed observations that after separation has started, the algorithm behaves within each separated region similar as under homogeneity as long as the increasing local neighborhood does not reach a distant region with similar values.
Additionally, for sufficiently large location bandwidths, the algorithm resulted for every test function in an adaptive weighting scheme, whose disjoint regions define a partition of the design space.
This indicates the immutability of the associated step function for sufficiently large location bandwidths.

Hence, the presented numerical results suggest the convergence of the algorithm, 
but we lack for a theoretical justification. 
There are three main reasons for this.
\begin{compactitem}
\item Each realization may yield another associated step function with slightly shifted steps. 
\item The improvement of the estimation quality during iteration is not ensured to be monotonic, neither for the non-adaptive nor for the adaptive estimates. 
Several other iterative methods, such as the expectation-maximization algorithm or the conjugate gradient method, rely on the minimization or maximization of a certain criterion.
This provides a monotonic improvement of some quality criterion, which ensures the convergence of the algorithm.
In contrast, the Propagation-Separation Approach considers an increasing local neighborhood, 
where unfavorable, newly included or stronger weighted observations may worsen the estimation quality in comparison to a previous iteration step.
\item The immutability of the associated step function for sufficiently large iteration steps requires the existence of some iteration step $k_0 < \infty$ such that the considered neighborhood equals the complete design, that is 
$\overline{w}_{ij}^{(k_0)} > 0$ for all $X_i, X_j \in \mathcal{X}$, and
%\[
$\lbrace X_j \in \mathcal{X}: s_{ij}^{(k_0)} \leq \lambda \rbrace = \lbrace X_j \in \mathcal{X}: s_{ij}^{(k)} \leq \lambda \rbrace $
for all $k > k_0$ and every $X_i \in \mathcal{X}$.
%	\qquad \text{ for all } k > k_0 \text{ and every } X_i \in \mathcal{X}.
%\]
\end{compactitem}

Let us consider the last reason in more detail.
We know from the definition of the statistical penalty, see Algorithm~\ref{algorithmMS} (page~\pageref{algorithmMS}), that a violation of the above condition can arise from
\begin{compactenum}
\item a reunion of previously separated regions due to a decrease of the factor~$\tilde{N}_i^{(k'-1)}$;
\item a reunion of previously separated regions due to a decrease of $\mathcal{KL}(\tilde{\theta}_i^{(k-1)}, \tilde{\theta}_j^{(k-1)})$;
\item a subsequent segmentation of a before created step due to an increase of the corresponding Kullback-Leibler divergence $\mathcal{KL}(\tilde{\theta}_i^{(k-1)}, \tilde{\theta}_j^{(k-1)})$;
\item a too strong intensification of the statistical penalty by the factor~$\tilde{N}_i^{(k-1)}$.
\end{compactenum}
We discuss these events case by case.

Recall that the non-adaptive sequence~$\lbrace \overline{N}_i^{(k)} \rbrace_{k=0}^{k^*}$ is monotonically increasing, whereas its adaptive counterpart~$\lbrace \tilde{N}_i^{(k)} \rbrace_{k=0}^{k^*}$ does not need to be monotonic.
Therefore, we propose a slight modification of the statistical penalty in Algorithm~\ref{algorithmMS}, setting
\[
	s_{ij}^{(k)} := \max_{k' \leq k} \tilde{N}_i^{(k'-1)} \mathcal{KL}(\tilde{\theta}_i^{(k-1)}, \tilde{\theta}_j^{(k-1)}).
\]
This modification preserves an already achieved adaptation quality.
As a consequence, it avoids that a design point switches all the time between two steps due to oscillation of the value of~$\tilde{N}_i^{(k-1)}$ during iteration.

A late segmentation and a reunion as described in~(2) and~(3) could be imposed by an appropriate upper bound of
\[
	\max \left\{ \mathcal{KL} \left( \tilde{\theta}_{i_1}^{(k)}, \tilde{\theta}_{i_2}^{(k)} \right): X_{i_1}, X_{i_2} \in \mathcal{H}_i^{(k)}  \right\}
\]
and a lower bound of
\[
	\min \left\{ \mathcal{KL} \left( \tilde{\theta}_{i_1}^{(k)}, \tilde{\theta}_{j_1}^{(k)} \right): X_{i_1} \in \mathcal{H}_i^{(k)}, X_{j_1} \in \mathcal{H}_j^{(k)} \neq \mathcal{H}_i^{(k)} \right\},
\]
where~$\mathcal{H}_i^{(k)}$ is as in Equation~\eqref{eq:partitionHik} (page~\pageref{eq:partitionHik}).
Due to the factor~$\varkappa$ in Lemma~\ref{lem:A1}, the corresponding discussion in Appendix~\ref{app:auxiliary}, and the missing monotonicity of the Kullback-Leibler divergences $\mathcal{KL}(\tilde{\theta}_i^{(k-1)}, \tilde{\theta}_j^{(k-1)})$ in $k>k_0$, this may lead to a criterion which is too restrictive to be satisfiable with $k_0 < \infty$.

However, the main impediment of a theoretical proof results from~(4).
The statistical penalty becomes more restrictive during iteration by the factor~$\tilde{N}_i^{(k)}$, but this factor is not guaranteed to be always appropriate.
For statistically independent observations~$\lbrace Y_j \rbrace_j$ with expected values~$\lbrace \theta_j \rbrace_j$ and variance~$ \sigma^2$, it may be explained as an upper bound of the by the non-adaptive estimator achieved variance reduction.
A generalization to the adaptive estimator may be prohibitive due to the randomness of the adaptive weights. 
Additionally, for other classes of probability distributions than the Gaussian one, the relation between the variance and the Kullback-Leibler divergence is complicated, and the variance may be heteroscedastic.
For instance, for exponentially distributed observations, the variance depends on the locally varying parameter~$\theta$.

Hence, we prefer to consider~$\tilde{N}_i^{(k)}$ as the achieved improvement of the estimation quality in terms of the Kullback-Leibler divergence. 
This is motivated by the Theorems~\ref{thm:PS21} and~\ref{thm:PS22} and the propagation condition, which yield with high probability and in case of sharp discontinuities for $\mathcal{KL} ( \tilde{\theta}_i^{(k)}, \mathcal{E} \tilde{\theta}_i^{(k)} )$ the rate of convergence~$\tilde{N}_i^{(k)}$, see Proposition~\ref{prop:propagation1}. 
Nevertheless, there remains a certain probability that the intensification of the statistical penalty is not justified. 
Furthermore, the mentioned propagation results do not generalize to the case of model misspecification.
All of these results are based on the propagation condition and this requires well separated regions. 
If the corresponding structural assumptions are violated, the impact of the adaptivity may change such that propagation cannot be ensured any more. 
In fact, model misspecification leads to a decrease of the probability for propagation.
Therefore, we may still observe propagation in practice, but its probability cannot be quantified as the established exponential bounds do not hold under model misspecification.
As a consequence, we cannot ensure neither the immutability of the associated step function nor the convergence of the simplified Propagation Separation Approach.

\subsection{(Where) do we need the memory step?}\label{sec:comparison}

As discussed in \citep[\S 5]{BecMat2013}, our approach is not constructed to provide asymptotic results, and the theoretical analysis is mainly restricted to piecewise bounded parameter functions.
Hence, we lose the general stability of estimates by \citet[Thm 5.7]{PoSp05}.
Nevertheless, we justified the essential properties for the simplified procedure, that is propagation and separation.
This emphasizes that both properties result from the adaptivity of the estimator, but not from the memory step.

From a practical point of view, the benefits of the memory step are still questionable.
In \S~\ref{sec:numerics-memory}, we illustrated the impact of the memory step for several test functions.
Using the default parameter choices of the \textbf{R}-package \pkg{aws} by \citet{aws}, we could not observe any effect of the memory step.
However, these choices are not arbitrary. 
The memory bandwidth was chosen in accordance with a former version of the propagation condition, and, indeed, we observed an increased risk of adaptation to noise for considerably smaller bandwidths.
On the one hand, this emphasizes the importance of a sufficiently large memory bandwidth to avoid adaptation to outliers. 
On the other hand, we got a smaller MAE by increasing the amount of aggregation, which slightly attenuated the formation of a step function during iteration.
In any case, we found the best results by restricting the maximal location bandwidths appropriately.
The omittance of the memory step provides a better interpretability of the procedure and, as a consequence, of the results since the memory step introduces additional interactions between the involved components, which are not fully understood yet.

\subsection{Future research}\label{sec:PS-future}

There are several topics for future research that arise from this article.
For instance, one could study the impact of the Kullback-Leibler divergence.
Especially for the justification of the homogeneous propagation condition in \citep[\S 4.1]{BecMat2013} and of the inhomogeneous propagation condition in Section~\ref{sec:inhomPCPractice}, we took advantage of its properties.
Are there other (possibly asymmetric) distance functions or f-divergences which provide similar results?

Here, we would like to concentrate on another question concerning the consequences of a violated structural assumption.
As indicated by our numerical results in Section~\ref{sec:simulations}, an appropriate stopping of the iterative procedure may reduce the resulting estimation bias considerably by avoiding the formation of a step function.
For the presented univariate examples, a choice by visual inspection seems to be promising.
In all observed cases, the iteration step where the formation of the step function started to dominate the smoothing result could be easily identified.
Additionally, we always observed a certain range of iteration steps where the estimation quality is very similar.

However, on more complicated test functions or for higher dimensional design spaces, an automatic choice of the maximal number of iterations is desired or even required.
In the context of local polynomial regression and locally weighted maximum likelihood estimation, there is a large amount of literature concerning the choice of the location bandwidth.
For instance, the maximal location bandwidth~$h^{(k^*)}$ could be chosen such that the non-adaptive estimator in Notation~\ref{not:algorithm} behaves well within regions without discontinuities. 
Then, assuming an appropriate choice of the adaptation bandwidth~$\lambda$, the simplified algorithm in Notation~\ref{algorithmMS} would yield similar results as non-adaptive smoothing within these regions, while smoothing among distinct regions would be avoided as sharp discontinuities could be detected by the adaptive weights. 
To evaluate the appropriateness of the different approaches for the Propagation-Separation Algorithm would form a promising research project for the future.
Alternatively, one could search for a criterion which takes advantage of the involved components of the method.
The evaluation of the behavior of the statistical penalty or of the sum of the adaptive weights could provide useful information about the iteration step where the formation of a step function negatively affects the smoothing results.

\appendix

\section{Some auxiliary results}\label{app:auxiliary}

We recall some previous results by \citet{PoSp05} and \citet{BecMat2013}.
The following lemma has been stated in \citet[Pages 339 \& 352]{PoSp05}, but see also \citet[Lem. 2.3]{BecMat2013}.

\begin{Lem}\label{lem:A1}
Under Assumption~\ref{A1} it holds the following.
\begin{compactenum}
\item The Fisher information satisfies $I(\theta) = C'(\theta)$ for all $\theta \in \Theta$.
\item For every compact  and convex subset $\Theta' \subseteq \Theta$,  there is a constant $\varkappa \geq 1$ such that 
\begin{equation}\label{eq:varkappa}
	\frac{I(\theta_1)}{I(\theta_2)} \leq \varkappa^2 \qquad
        \text{ for all } \theta_1, \theta_2 \in \Theta'.
\end{equation}
\item The Kullback-Leibler divergence is convex with respect to the first argument.
Additionally, it has an explicit representation,
\begin{eqnarray}\label{eq:KL}
	\mathcal{KL} \left( \theta, \theta' \right) 
	&=& \theta \left[ C(\theta) - C(\theta') \right] - \left[ B(\theta) - B(\theta') \right], \quad \theta, \theta' \in \Theta.
\end{eqnarray}
\end{compactenum}
\end{Lem}

Equation~\eqref{eq:varkappa} allows the following notations.

\begin{Not}\label{not:varkappa}
For every compact  and convex subset $\Theta' \subseteq \Theta$, we set
\[
	\varkappa := \max \lbrace I(\theta_1) / I(\theta_2): \theta_1, \theta_2 \in \Theta' \rbrace \geq 1
	\quad \text{ and } \quad \Theta' := \Theta_{\varkappa}.
\]
Vice versa, for every constant $\varkappa \geq 1$, we use the notation $\Theta' := \Theta_{\varkappa}$ for any compact and convex set $\Theta' \subseteq \Theta$ which satisfies Equation~\eqref{eq:varkappa}.
\end{Not}

In this study, we usually require the set~$\Theta_{\varkappa}$ to be sufficiently large such that $\theta(X_i) \in \Theta_{\varkappa}$ for all $i \in \lbrace 1,...,n \rbrace$. 
Its precise choice will be specified where necessary.

\begin{Lem}\label{lem:PS52}
Suppose Assumption~\ref{A1} and let~$\Theta_{\varkappa} \subseteq \Theta$ and~$\varkappa \geq 1$ be as in Notation~\ref{not:varkappa}.
For any sequence~$\theta_0, \theta_1,..., \theta_m \in \Theta_{\varkappa}$, it holds
\[
	\mathcal{KL}^{1/2} \left( \theta_0, \theta_m \right)
	\leq \varkappa \, \sum_{l=1}^m \mathcal{KL}^{1/2} \left( \theta_{l-1}, \theta_l \right).
\]
\end{Lem}

For the proof, we refer the reader to \citet[Lem. 5.2]{PoSp05}.
Here, we recall some details concerning the applicability of the technical Lemma~\ref{lem:PS52} and the related Equation~\eqref{eq:varkappa}, see \citep[App. A]{BecMat2013}.
For several results, we apply Equation~\eqref{eq:varkappa} and Lemma~\ref{lem:PS52} not only with respect to the true parameters~$\lbrace \theta_i \rbrace_i$, but as well with respect to the transformed observations~$\lbrace T(Y_i) \rbrace_i$ or the associated estimates~$\lbrace \tilde{\theta}_i^{(k)} \rbrace_i$, $k \in \lbrace 0,..., k^* \rbrace$. 
Therefore, we restrict our analysis to the favorable realizations $\lbrace T(Y_i) \in \Theta_{\varkappa} \text{ for all } i \rbrace$, and we quantify the probability of its complementary set. 
For every~$\varkappa$, we use the most convenient choice of the set~$\Theta_{\varkappa}$. 
Furthermore, we restrict the range of~$\theta(.)$ by a subset $\Theta^* \subseteq \Theta$. 

\begin{Not}\label{not:pKappa}
We fix a subset $\Theta^* \subseteq \Theta$ and a constant $\varphi_0 \geq 0$. 
Then, we recall Notation~\ref{not:varkappa}, and we choose $\varkappa \geq 1$ sufficiently large such that $\Theta^* \subseteq \Theta_{\varkappa}$.
We define the function $\mathfrak{p}_{\varkappa}: \left( \Theta^* \right)^n \to [0,1]$ by
\[
	\mathfrak{p}_{\varkappa} \left( \lbrace \theta_i \rbrace_{i=1}^n \right) 
	:= \inf \lbrace \mathbb{P} \left( \exists \, i \in \lbrace 1,...,n \rbrace: \, T(Y_i) \notin \Theta_{\varkappa} \right): 
	Y_i \sim \mathbb{P}_{\theta_i}, 
	\lbrace \theta_i \rbrace_{i=1}^n \in (\Theta_{\varkappa})^n \rbrace.
\]
The worst choice of $\lbrace \theta_i \rbrace_{i=1}^n \in (\Theta^*)^n$ with bounded Kullback-Leibler divergence yields the probability
\begin{eqnarray*}
	p_{\varkappa} &:=& \sup \lbrace \, \mathfrak{p}_{\varkappa} \left( \lbrace \theta_i \rbrace_{i=1}^n \right): \, 
	\lbrace \theta_i \rbrace_{i=1}^n \in (\Theta^*)^n \text{ and } \max_{i,j} \mathcal{KL} (\theta_i, \theta_j) \leq \varphi_0^2 \rbrace.
\end{eqnarray*}
\end{Not}

The probability~$p_{\varkappa}$ decreases with increasing $\varkappa \geq 1$. 

\begin{Ex}\label{ex:pKappa}\hspace{1 pt}
\begin{compactitem}
\item For Gaussian and log-normal distributed observations with $\mathcal{P} = \lbrace \mathcal{N}(\theta, \sigma^2) \rbrace_{\theta \in \Theta}$ and $\mathcal{P} = \lbrace \text{log} \, \mathcal{N}(\theta, \sigma^2) \rbrace_{\theta \in \Theta}$, respectively, it holds $I(\theta) = 1/ \sigma^2$, leading to $\varkappa = 1$.
In this case, Equation~\eqref{eq:varkappa} and Lemma~\ref{lem:PS52} hold for every subset~$\Theta' \subseteq \Theta$ without the restriction to compact sets, and we get $p_{\varkappa} = 0$. 
This is the optimal scenario.
\item For the Gamma, Erlang, scaled chi-squared, exponential, Rayleigh, Weibull, and Pareto distributions, it holds after reparametrization $I(\theta) = 1/ \theta^2$. 
In this case~$\varkappa$ and~$p_{\varkappa}$ become large. 
However, the effective values of~$\varkappa$ and~$p_{\varkappa}$ may be much smaller than the global ones,
which attenuates the consequences in practice.
\end{compactitem}
\end{Ex}

Next we recall an exponential bound by \citet[Thm. 6.1]{PoSp05}, from which the Theorems~\ref{thm:PS21} and~\ref{thm:PS22} follow as special cases.

\begin{Thm}\label{thm:PS61}
Suppose Assumption~\ref{A1}, and reparametrize~$v := C(\theta)$ and~$ D(v) := B(\theta)$.
Furthermore, let $\overline{W}_i := \lbrace \overline{w}_{ij} \rbrace_{j=1}^n \in [0,1]^n$ be a weighting scheme, 
and consider the corresponding non-adaptive estimator~$\overline{\theta}_i$ in Notation~\ref{not:algorithm} and its expectation $\breve{\theta}_i := \mathbb{E} \overline{\theta}_i = \sum_{j} \overline{w}_{ij} \theta_j / \overline{N}_i$. 
We set $q(u|v) := \mathcal{KL} (v, v+u)$ and define, for a given constant~$z \geq 0$ and~$\breve{v}_i = C(\breve{\theta}_i)$, the set
\[
	\mathcal{U} \left( \overline{W}_i, z \right) := \left\{ u \in \mathbb{R}: \int_0^u x D''(\breve{v}_i+ x) dx = z / \overline{N}_i \right\},
\]
where it holds $\int_0^u x D''(\breve{v}_i+ x) dx = \mathcal{KL} (\breve{v}_i + u, \breve{v}_i)$
Finally, we assume the existence of some constant~$\alpha \geq 0$ such that
\begin{equation}\label{eq:expBound}
	q(\mu u \overline{w}_{ij} | v_j) \leq (1 + \alpha) \mu^2 \overline{w}_{ij} q(u|\breve{v}_i), \qquad j=1,...,n,
\end{equation}
for $\mu := (1 + \alpha)^{-1} \in (0,1]$ and all $u \in \mathcal{U} \left( \overline{W}_i, z \right)$. 
Then, we get
\[
	\mathbb{P} \left( \overline{N}_i \, \mathcal{KL} \left( \overline{\theta}_i, \mathbb{E} \overline{\theta}_i \right) > z \right)
	\leq 2 e^{-z/(1+ \alpha)}.
\]
\end{Thm}

\begin{Rem}
\citet{PoSp05} assumed the sufficient statistic~$T$ in Assumption~\ref{A1} to be the identity map.
Fortunately, Theorem~\ref{thm:PS61} depends on the probability distribution and consequently on~$T$ via the Kullback-Leibler divergence only. 
This ensures with Lemma~\ref{lem:A1}~(3) that the choice of~$T$ does not have any effect, and the original result remains valid.
\end{Rem}

The next result can be found in \citep[Thm. 2.1]{PoSp05}.

\begin{Thm}\label{thm:PS21}
Let Assumption~\ref{A1} be satisfied, and presume a parametric model, $\theta(.) \equiv \theta$. 
Then, for each $i \in \lbrace 1,...,n \rbrace$ and every weighting scheme~$\overline{W}_i := \lbrace \overline{w}_{ij} \rbrace_{j=1}^n \in [0,1]^n$,
we get
\[
	\mathbb{P} \left( \overline{N}_i \mathcal{KL}(\overline{\theta}_i, \theta) > z \right) \leq 2 e^{-z} \quad \text{ for all } z > 0
\]
with~$\overline{N}_i$ and~$\overline{\theta}_i$ as in Notation~\ref{not:algorithm}.
\end{Thm}

Next we extend Theorem~\ref{thm:PS21} to parameter functions with bounded variability.
In the corresponding proof, \citet{PoSp05} used Equation~\eqref{eq:varkappa} (page~\pageref{eq:varkappa}).
Although not stated in \citep[Thm. 2.2]{PoSp05}, this requires a restriction to the favorable realizations, which may worsen the result.
In order to quantify the probability of the complementary set, we proceed in an analogous manner as in Notation~\ref{not:pKappa}.
We consider a different set of realizations, whose definition will be motivated in the proof of Theorem~\ref{thm:PS22} (page~\pageref{proof:PS22}).

\begin{Not}\label{not:pKappa2}
Recall Notation~\ref{not:varkappa}.
We fix a subset $\Theta^* \subseteq \Theta$ and a constant $\varphi_0 \geq 0$. 
Let $\varkappa \geq 1$ be sufficiently large such that $ \Theta^* \subseteq \Theta_{\varkappa}$.
Then, for every $i \in \lbrace 1,..., n \rbrace$, we consider the non-adaptive estimator~$\overline{\theta}_i$ in Notation~\ref{not:algorithm}
with weighting scheme $\overline{W}_i := \lbrace \overline{w}_{ij} \rbrace_{j=1}^n \in [0,1]^n$.
The function $\breve{\mathfrak{p}}_{\varkappa}: \left( \Theta^* \right)^n \to [0,1]$ is given by
\begin{eqnarray*}
	\breve{\mathfrak{p}}_{\varkappa} \left( \lbrace \theta_i \rbrace_{i=1}^n \right) 
	& := \inf \lbrace & \mathbb{P} \left( \exists \, i,j \in \lbrace 1,...,n \rbrace: \, C^{-1} [ C(\theta_j) + C(\overline{\theta}_i) - C(\mathbb{E} \overline{\theta}_i) ] \notin \Theta_{\varkappa} \right): \\
	&& Y_i \sim \mathbb{P}_{\theta_i}, 
	\lbrace \theta_i \rbrace_{i=1}^n \in (\Theta_{\varkappa})^n \: \rbrace.
\end{eqnarray*}
Furthermore, we consider the worst choice of $\lbrace \theta_i \rbrace_{i=1}^n \in (\Theta^*)^n$ with bounded Kullback-Leibler divergence via
\begin{eqnarray*}
	\breve{p}_{\varkappa} &:=& \sup \left\{ \, \breve{\mathfrak{p}}_{\varkappa} \left( \lbrace \theta_i \rbrace_{i=1}^n \right): \, 
	\lbrace \theta_i \rbrace_{i=1}^n \in (\Theta^*)^n \text{ and } \max_{i,j} \mathcal{KL} (\theta_i, \theta_j) \leq \varphi_0^2 \right\}.
\end{eqnarray*}
For every $i \in \lbrace 1,...,n \rbrace$, let the weighting scheme $\overline{W}_i := \lbrace \overline{w}_{ij} \rbrace_{j=1}^n \in [0,1]^n$ be given as $\overline{w}_{ii} = 1$ and $\overline{w}_{ij} = 0$ for all $j \neq i$.
Then, it holds $\overline{\theta}_i = T(Y_i)$ for every~$i$, and we set $\breve{p}_{\varkappa,0} := \breve{p}_{\varkappa}$ in order to distinguish the specific weighting scheme.
\end{Not}

\begin{Ex}
For Gaussian and log-normal distributed observations, it holds $\breve{p}_{\varkappa}=0$ since $\varkappa = 1$ for every set $\Theta_{\varkappa} \subseteq \Theta$. 
For the Gamma and its related distributions, the probability~$\breve{p}_{\varkappa}$ may be large, and it increases with decreasing values of~$\varkappa$ as well as with increasing sample sizes.
\end{Ex}

The probabilities~$\breve{p}_{\varkappa,0}$ in Notation~\ref{not:pKappa2} and~$p_{\varkappa}$ in Notation~\ref{not:pKappa} are closely related.

\begin{Cor}\label{cor:pKappa}
Suppose Assumption~\ref{A1} and the setting of Notation~\ref{not:pKappa2}.
Then, it holds 
\begin{eqnarray*}
	\breve{\Omega}_{\varkappa}
	:= \bigcap_{i, j=1}^n \left\{ C^{-1} \left[ C(\theta_j) + C(T(Y_i)) - C(\theta_i) \right] \in \Theta_{\varkappa} \right\}
	\subseteq \bigcap_{i=1}^n \left\{ T(Y_i) \in \Theta_{\varkappa} \right\}
	=: \Omega_{\varkappa},
\end{eqnarray*}
and, as a consequence, we get $\breve{p}_{\varkappa,0} \geq p_{\varkappa}$, where~$p_{\varkappa}$ is as in Notation~\ref{not:pKappa}.
\end{Cor}

\begin{Thm}\label{thm:PS22}
Suppose Assumption~\ref{A1}, and fix a subset $\Theta^* \subseteq \Theta$ and a constant $\varphi_0 \geq 0$ such that
$\lbrace \theta_i \rbrace_{i=1}^n \in (\Theta^*)^n$ and $\max_{i,j} \mathcal{KL} (\theta_i, \theta_j) \leq \varphi_0^2$.
Moreover, recall Notation~\ref{not:varkappa} and let $\varkappa \geq 1$ be sufficiently large such that $\Theta^* \subseteq \Theta_{\kappa}$.
Finally, let $\overline{W}_i := \lbrace \overline{w}_{ij} \rbrace_{j=1}^n \in [0,1]^n$ denote a weighting scheme, and recall the corresponding quantities~$\overline{\theta}_i$ and~$\overline{N}_i$ in Notation~\ref{not:algorithm}.
Then, for each $i \in \lbrace 1,...,n \rbrace$ and every $z>0$,
it holds
\[
	\mathbb{P} \left( \overline{N}_i \mathcal{KL} (\overline{\theta}_i, \mathbb{E} \overline{\theta}_i) > z \right) 
	\leq 2 e^{-z/\varkappa^2} + \breve{p}_{\varkappa},
\]
where~$\breve{p}_{\varkappa}$ is as in Notation~\ref{not:pKappa2}.
\end{Thm}

We give the proof in Appendix~\ref{app:proofs} in order to clarify the appearance of the probability~$\breve{p}_{\varkappa}$.
Finally, we recall the separation property, which can be found in \citep[Thm. 5.9]{PoSp05} and \citep[Prop. 3.2]{BecMat2013}.

\begin{Prop}\label{prop:separationMS}
Suppose Assumption~\ref{A1}, and consider two design points $X_{i_1}, X_{i_2} \in \mathcal{X}$.
Assume that the realization at hand satisfies at these points in iteration step~$k$ the estimation accuracy $\mathcal{KL}(\tilde{\theta}_{i_m}^{(k)}, \theta_{i_m}) \leq z_m^{(k)} := z / \overline{N}_{i_m}^{(k)}$ with some constant~$z > 0$ 
and $\theta_{i_m}, \tilde{\theta}_{i_m}^{(k)} \in \Theta_{\varkappa}$, $m=1,2$, for $\varkappa \geq 1$ fixed and~$\Theta_{\varkappa}$ as in Notation~\ref{not:varkappa}. 
If additionally
\begin{equation}\label{eq:varphi1}
	\mathcal{KL}^{1/2} \left( \theta_{i_1}, \theta_{i_2} \right) > \varkappa \left( \sqrt{\lambda /\tilde{N}_{i_1}^{(k)}} + \sqrt{z_1^{(k)}} + \sqrt{z_2^{(k)}} \right),
\end{equation}
then we get~$\tilde{w}_{i_1 i_2}^{(k+1)} = 0$.
\end{Prop}

\section{Proofs}\label{app:proofs}

\begin{proof}[Proof of Proposition~\ref{prop:propagation1}]
The inhomogeneous propagation condition yields the monotonicity of the function~$\hat{\mathfrak{Z}}_{\lambda} \left(k, p; \lbrace \theta_j \rbrace_{j=1}^n, i \right)$ in $k \leq k_0$ for all~$p \in (\epsilon, 1)$ and every $i \in \lbrace 1,...,n \rbrace$. 
This implies Equation~\eqref{eq:propCondInhom2}. 
We turn to Equation~\eqref{eq:propCondInhom1}, and we consider the event~$\breve{\Omega}_{\varkappa}$ in Corollary~\ref{cor:pKappa}. The adaptive estimator is defined as a weighted mean of the transformed observations. Therefore, for all $k \in \lbrace 0, ..., k^* \rbrace$, we get
$\breve{\Omega}_{\varkappa} \subseteq \lbrace \tilde{\theta}_i^{(k)} \in \Theta_{\varkappa}, i \in \lbrace 1,..., n \rbrace \rbrace$.
Then, we use the convexity of the Kullback-Leibler divergence with respect to the first argument, see Lemma~\ref{lem:A1}. 
Denoting the complement of the set~$M$ by~$M^c$, it follows from the initialization of Algorithm~\ref{algorithmMS} (page~\pageref{algorithmMS}) with~$\tilde{\theta}_i^{(0)} = \overline{\theta}_i^{(0)}$ that
\begin{eqnarray*}
	&& \mathbb{P} \left( \overline{N}_i^{(k)} \mathcal{KL} \left( \tilde{\theta}_i^{(k)}, \theta_i \right) > z \right) \\
	&\overset{\text{Lem.~\ref{lem:PS52}}}{\leq} & \mathbb{P} \left( \left\{ \varkappa^2 \overline{N}_i^{(k)} \left[ \mathcal{KL}^{1/2} \left( \tilde{\theta}_i^{(k)}, \mathcal{E} \tilde{\theta}_i^{(k)} \right) + \mathcal{KL}^{1/2} \left( \mathcal{E} \tilde{\theta}_i^{(k)}, \theta_i \right) \right]^2 > z \right\} \cap \breve{\Omega}_{\varkappa} \right) + \mathbb{P} \left( \breve{\Omega}_{\varkappa}^c \right) \\
	&\overset{\text{Lem.~\ref{lem:A1}}}{\leq}& \mathbb{P} \left( \left\{ \overline{N}_i^{(k)} \mathcal{KL} \left( \tilde{\theta}_i^{(k)}, \mathcal{E} \tilde{\theta}_i^{(k)} \right) > \left[ \sqrt{z} / \varkappa - \varphi \right]^2 \right\} \cap \breve{\Omega}_{\varkappa} \right) + \breve{p}_{\varkappa,0} \\
	&\overset{\text{Eq.~\eqref{eq:propCondInhom2}}}{\leq}&
	\max \left\{ \mathbb{P} \left( \left\{ \overline{N}_i^{(0)} \mathcal{KL} \left( \overline{\theta}_i^{(0)}, \mathbb{E} \overline{\theta}_i^{(0)} \right) > \left[ \sqrt{z} / \varkappa - \varphi \right]^2 \right\} \cap \breve{\Omega}_{\varkappa} \right), \epsilon \right\}  + \breve{p}_{\varkappa,0}\\
	&\overset{\text{Thm.~\ref{thm:PS22}}}{\leq}& \max \left\{ 2 e^{-\left[ \sqrt{z} / \varkappa - \varphi \right]^2 / \varkappa^2}, \epsilon \right\} + \breve{p}_{\varkappa,0}
\end{eqnarray*}
since the event~$\breve{\Omega}_{\varkappa}$ is independent of the iteration step~$k$.
\end{proof}

\begin{proof}[Proof of Theorem~\ref{thm:propagation2}]
Recall the event~$\breve{\Omega}_{\varkappa}$
in Corollary~\ref{cor:pKappa}, and let~$M^c$ denote the complement of the set~$M$. 
We construct a disjoint union
\begin{eqnarray*}
	\left[ \mathcal{B}^{(k_0)}(z) \right]^c 
	& = & \bigcup_{k = 0}^{k_0} \left( \left[ \mathcal{B}^{(k)}(z) \right]^c \cap \left[ \bigcap_{k'=0}^{k - 1} \mathcal{B}^{(k')}(z) \right] \right),
\end{eqnarray*}
where we set $\underset{k'=0}{\overset{k - 1}{\bigcap}} \mathcal{B}^{(k')}(z) := \Omega$ if $k=0$.
Then, we get
\begin{eqnarray}
	&&\mathbb{P} \left( \mathcal{B}^{(k_0)}(z) \vert M^{(k_0)}(z) \right)  \label{eq:Bkmu} \\
	&\geq & 1 - \frac{\left[ \breve{p}_{\varkappa,0}  
	+ \sum_{k = 0}^{k_0} \mathbb{P} \left( M^{(k)}(z) \cap \breve{\Omega}_{\varkappa} \cap \left[ \mathcal{B}^{(k)}(z) \right]^c 
	\cap \left[ \bigcap_{k'=0}^{k - 1} \mathcal{B}^{(k')}(z) \right] \right) \right] }{\mathbb{P} \left( M^{(k_0)}(z) \right) }, \nonumber
\end{eqnarray}
where we used that $M^{(k_0)}(z) \subseteq M^{(k)}(z)$ for $k \leq k_0$.
The choice of~$h^{(0)}$ ensures, for every $i \in \lbrace 1,...,n \rbrace$, that $U_i^{(0)} \setminus \mathcal{V}_{i} = \emptyset$.
Moreover, it holds $\tilde{\theta}_{i}^{(0)} = \overline{\theta}_{i}^{(0)}$ by the initialization of Algorithm~\ref{algorithmMS} (page~\pageref{algorithmMS}), and it follows in the same manner as in the proof of Proposition~\ref{prop:propagation1} that
\begin{eqnarray}
	\mathbb{P} \left( M^{(0)}(z) \cap \breve{\Omega}_{\varkappa} \cap \left[ \mathcal{B}^{(0)} (z) \right]^c \right) \label{eq:B0mu} 
	&\overset{\overline{n}_{i}^{(0)} = \overline{N}_{i}^{(0)}}{\leq} & 
	n \cdot \mathbb{P} \left( \overline{N}_{i}^{(0)} \mathcal{KL} ( \overline{\theta}_{i}^{(0)}, \theta_{i} ) > z \right) \\
	& \leq & 2 \, n e^{-\left[ \sqrt{z} / \varkappa - \varphi \right]^2 / \varkappa^2}. \nonumber
\end{eqnarray}
By definition of the events~$\mathcal{B}^{(k)}(z)$, $M^{(k)}(z)$, and~$\breve{\Omega}_{\varkappa}$,
the conditions of Proposition~\ref{prop:separationMS} are satisfied on the intersection
\[
	M^{(k)}(z) \cap \breve{\Omega}_{\varkappa} \cap \left[ \bigcap_{k'=0}^{k - 1} \mathcal{B}^{(k')}(z) \right]
\]
for all $k \in \lbrace 1,..., k_0 \rbrace$.
There, it follows that~$\tilde{w}_{ij}^{(k)} = 0$ for all $ X_{j} \notin U_{i}^{(k)} \cap \mathcal{V}_{i}$. 
Hence, smoothing is restricted to the neighborhood~$\mathcal{V}_{i}$, and 
$\mathbb{E} \left[ T(Y_j) \right] = \theta_{i}$ for every~$X_j$ with $\tilde{w}_{ij}^{(k)} > 0$. 
Then, we get with Proposition~\ref{prop:propagation1} that
\begin{eqnarray}
	&&\mathbb{P} \left( \lbrace \overline{n}_{i}^{(k)} \mathcal{KL} \left( \tilde{\theta}_{i}^{(k)}, \theta_{i} \right) > z \rbrace \cap M^{(k)}(z) \cap \breve{\Omega}_{\varkappa} \cap \left[ \bigcap_{k'=0}^{k - 1} \mathcal{B}^{(k')}(z) \right] \right) \nonumber \\
	&& \qquad 
	\leq \max \left\{ 2 e^{-\left[ \sqrt{z} / \varkappa - \varphi \right]^2 / \varkappa^2}, \epsilon \right\} \label{eq:locPropNeuMS3}
\end{eqnarray}
 for all~$k \in \lbrace 1,..., k_0 \rbrace $. 
Finally, Equations~\eqref{eq:Bkmu}, \eqref{eq:B0mu}, and~\eqref{eq:locPropNeuMS3} lead to
\begin{eqnarray*}
	\mathbb{P} \left( \mathcal{B}^{(k_0)}(z) \vert M^{(k_0)}(z) \right) 
	\geq 1 - \frac{ \breve{p}_{\varkappa,0} + (k_0+1) \max \left\{ 2 n e^{-\left[ \sqrt{z} / \varkappa - \varphi \right]^2 / \varkappa^2}, n \epsilon \right\} }{ \mathbb{P} \left( M^{(k_0)}(z) \right) }.
\end{eqnarray*}
This terminates the proof.
\end{proof}

\begin{proof}[Proof of Lemma~\ref{lem:funcP1}]
It holds
\begin{compactitem}
\item $T(Y) = \ln(Y) \sim \mathcal{N} (\mu, \sigma^2)$ if $Y \sim \mathrm{log} \mathcal{N} (\mu, \sigma^2)$ with $\sigma^2 > 0$; 
\item $T(Y) = Y^2 \sim \mathrm{Exp} \left( \frac{1}{2 \theta^2} \right)$ if $Y \sim \mathrm{Rayleigh} (\theta)$;
\item $T(Y) = Y^k\sim \mathrm{Exp} \left( \frac{1}{\theta^k} \right)$ if $Y \sim \mathrm{Weibull} (\theta, k)$ with $k > 0$;
\item $T(Y) = \ln \left( y / x_m \right) \sim \mathrm{Exp} \left( \frac{1}{\theta} \right)$ if $Y \sim \mathrm{Pareto} (x_m, \frac{1}{\theta})$ with $x_m \geq 1$.
\end{compactitem}
Hence, in each of these cases, the non-adaptive estimator follows the same distribution as for Gaussian or exponentially distributed observations. 
Additionally, it holds $\mathrm{Exp} (1/\theta) = \Gamma (1, \theta)$, $\mathrm{Erlang}(k, \theta) = \Gamma(k, \theta)$, and $Y \sim \Gamma \left( k/2, 2 \theta / k \right)$ if $kY/ \theta \sim \chi^2(k) = \Gamma \left(k/2, 2 \right) $, where $k \in \mathbb{N}$. 
Since the associated Kullback-Leibler divergences coincide, it suffices to show the assertion for the Gaussian and the Gamma distribution, which satisfy Assumption~\ref{A1} with $T = Id$. 
By Lemma~\ref{lem:A1}, the function $g_{\theta}(y) = \mathcal{KL} \left( y, \theta \right)$ fulfills, for every $y \in \mathcal{Y}$, that
\begin{equation}\label{eq:KL-monotonicity}
	\frac{d g_{\theta}}{dy} (y) = \left[ C(y) - C(\theta) \right],
\end{equation}
where we used that $B'(y) = y C'(y)$.
Due to the strict monotonicity of the function~$C$, $\lbrace Y > \theta \rbrace$ and $\lbrace Y \leq \theta \rbrace$ restrict the random variable $\mathcal{KL} \left( y, \theta \right)$ to its regions of monotonicity. 
On each of these regions the assertion follows from \citep[Ex. 4.3]{BecMat2013}.
\end{proof}

\begin{proof}[Proof of Lemma~\ref{lem:funcP2}]
The parametric family of probability distributions~$\mathcal{P}$ is known, and it characterizes, for every $\vartheta \in \Theta$, the function $g_{\vartheta}: \mathcal{Y} \to [0, \infty)$ given by $g_{\vartheta}(y) = \mathcal{KL} \left( T(y), \vartheta \right)$.
Moreover, the probability distribution of the observations $\mathcal{Y}_i \overset{\text{iid}}{\sim} \mathbb{P}_{\vartheta} \in \mathcal{P}$ determines the probability distribution of the random quantities~$\tilde{N}_i^{(k)}$ and~$\tilde{\vartheta}_i^{(k)}$ and consequently of the random variables $\lbrace \tilde{N}_i^{(k)} \mathcal{KL} ( \tilde{\vartheta}_i^{(k)}, \vartheta ) \rbrace_i$, on which the function~$\mathfrak{Z}_{\lambda}$ is based. 
The function~$\mathfrak{Z}_{\lambda}$ is invariant with respect to the parameter $\vartheta \in \Theta$ if and only if the homogeneous propagation condition is invariant with respect to~$\vartheta$.
Therefore, it suffices to show that, for every $\vartheta \in \Theta$, the functions~$p^{(l)}_{\vartheta}$ with $l=1,2,3$ allow exact reconstruction of~$\mathbb{P}_{\vartheta}$ via the inverse of~$g_{\vartheta}$.

If the sufficient statistic in Assumption~\ref{A1} satisfies $T = Id$, then it follows from Equation~\eqref{eq:KL-monotonicity} that the inverse~$g_{\vartheta}^{-1}$ has exactly one solution on $\lbrace y < \vartheta \rbrace$ and $\lbrace y > \vartheta \rbrace$, respectively, and it holds $g_{\vartheta}^{-1}(0) = \vartheta$.
If $T \neq Id$, then we get
\begin{eqnarray*}
	\frac{d g_{\vartheta}}{dy} (y) 
	&=& T'(y) \left[ C(T(y)) - C(\vartheta) \right],
\end{eqnarray*}
where the assumed strict monotonicity of~$T$ leads again to the regions of monotonicity $\lbrace T(y) > \vartheta \rbrace$ and $\lbrace T(y) \leq \vartheta \rbrace$. 
Furthermore, knowledge of~$p^{(l)}_{\vartheta}$ with $l=1,2,3$ yields knowledge of $\mathbb{P}(\lbrace T(\mathcal{Y}_i) > \vartheta \rbrace \cap \lbrace \mathcal{KL}(T(\mathcal{Y}_i), \vartheta) \leq z \rbrace)$, $i \in \lbrace 1,...,n \rbrace$.
Therefore, we can reconstruct~$\mathbb{P}_{\vartheta}$ for every~$\vartheta$ from~$p^{(l)}_{\vartheta}$, $l=1,2,3$, which leads to the assertion.
\end{proof}

For the proof of Proposition~\ref{prop:inhomPC}, we recall the following basic result.

\begin{Cor}\label{cor:basicBounds}
Let $a_1, a_2, b_1, b_2 \in \mathbb{R}$ satisfy $(a_1 - b_1) \cdot (a_2 - b_2) \geq 0$.
Then it holds
\[
	\vert a_1 - a_2 \vert \leq \left| \vert a_1 - b_1 \vert - \vert a_2 - b_2 \vert \right| + \vert b_1 - b_2 \vert.
\]
\end{Cor}

\begin{proof}[Proof of Proposition~\ref{prop:inhomPC}]
Lemma~\ref{lem:PS52} provides on~$M_0$, for all $i,j \in \lbrace 1,..., n \rbrace$, that
\begin{eqnarray}\label{eq:KL-M2}
	\mathcal{KL} (Y_i, Y_j)
	\leq \varkappa^2 \left[ \mathcal{KL}^{1/2} \left( Y_i, \theta_i \right) + \mathcal{KL}^{1/2} \left( \theta_i, \theta_j \right) + \mathcal{KL}^{1/2} \left( Y_j, \theta_j \right) \right]^2.
\end{eqnarray}
On a certain set of realizations this upper bound can be improved by Corollary~\ref{cor:basicBounds}.
For this purpose, we distinguish the following sets
\begin{eqnarray}\label{eq:M1-M2}
	\quad M_1 := \lbrace (Y_i - \theta_i) \cdot (Y_j - \theta_j) \geq 0 \rbrace 
	\, \text{ and } \, 
	M_2 := \lbrace (Y_i - \theta_i) \cdot (Y_j - \theta_j) < 0 \rbrace 
\end{eqnarray}
and analogously in the homogeneous setting
\begin{eqnarray}\label{eq:M3-M4}
	\quad M_3 := \lbrace (\mathcal{Y}_i - \vartheta) \cdot (\mathcal{Y}_j - \vartheta) \geq 0 \rbrace 
	\, \text{ and } \,
	M_4 := \lbrace (\mathcal{Y}_i - \vartheta) \cdot (\mathcal{Y}_j - \vartheta) < 0 \rbrace.
\end{eqnarray}

Now we separately reduce the Kullback-Leibler divergence $\mathcal{KL} (Y_i, Y_j)$ on~$M_1$ and on~$M_2$ to appropriate terms which only depend on the divergences $\mathcal{KL} (Y_i, \theta_i)$ and $\mathcal{KL} (Y_j, \theta_j)$.
Then, the invariance of the functions~$p^{(l)}_{\theta}$, $l=1,2,3$, with respect to the parameter~$\theta$ allows a comparison with the homogeneous analogs, namely $\mathcal{KL} (\mathcal{Y}_i, \vartheta)$ and $\mathcal{KL} (\mathcal{Y}_j, \vartheta)$.
Due to the separate handling of the realizations on~$M_1$ and on~$M_2$, the resulting formulas can be reduced to the divergence $\mathcal{KL} \left( \mathcal{Y}_i, \mathcal{Y}_j \right)$, which will lead to the assertion.

On the set~$M_2$, we just use the upper bound~\eqref{eq:KL-M2}.
For the set~$M_1$, we get from Taylor's Theorem for all $\theta_1, \theta_2 \in \Theta$ the existence of a parameter $\theta^{\ast} \in \Theta$ between~$\theta_1$ and~$\theta_2$ which satisfies
\begin{equation}\label{eq:KL-Taylor}
	\mathcal{KL} \left( \theta_1, \theta_2 \right) 
	= \tfrac{1}{2 I(\theta^{\ast})} \left[ C(\theta_1) - C(\theta_2) \right]^2.
\end{equation}
Therefore, on~$M_1 \cap M_0$ it follows from Corollary~\ref{cor:basicBounds} and the monotonicity of the function~$C$ that
\begin{eqnarray}
	\mathcal{KL} (Y_i, Y_j) 
	& \leq & \frac{1}{2 I(\theta^{\ast})} \left[ \left| \vert C(Y_i) - C(\theta_i) \vert - \vert C(Y_j) - C(\theta_j) \vert \right| + \vert C(\theta_i) - C(\theta_j) \vert \right]^2 \nonumber \\
	& \overset{\text{Eq.~\eqref{eq:varkappa}}}{\leq} & \varkappa^2 \left[ \left| \mathcal{KL}^{1/2} \left( Y_i, \theta_i \right) - \mathcal{KL}^{1/2} \left( Y_j, \theta_j \right) \right| + \mathcal{KL}^{1/2} \left(  \theta_i, \theta_j \right) \right]^2. \label{eq:KL-M1}
\end{eqnarray}
Then, using the invariance of~$p^{(l)}_{\theta}$, $l=1,2,3$, with respect to the parameter~$\theta$ and $\mathcal{KL} (\theta_i, \theta_j) \leq \varphi_0^2$ for all $i,j \in \lbrace 1,...,n \rbrace$, we can deduce that
\begin{eqnarray*}
	&&\mathbb{P} \left( \lbrace \mathcal{KL} (Y_i, Y_j) > z \rbrace \cap M_0 \right) \\
%	& = & \mathbb{P} \left( M_1 \cap M_0 \cap \lbrace \mathcal{KL} (Y_i, Y_j) > z \rbrace \right)
%	 + \mathbb{P} \left( M_2 \cap M_0 \cap \lbrace \mathcal{KL} (Y_i, Y_j) > z \rbrace \right) \\
	& \leq & \mathbb{P}\left( M_1 \cap M_0 \cap \left\{ \left| \mathcal{KL}^{1/2} \left( Y_i, \theta_i \right) - \mathcal{KL}^{1/2} \left( Y_j, \theta_j \right) \right|^2 > \left[ \sqrt{z}/ \varkappa - \varphi_0 \right]^2 \right\} \right) \\
	&& + \, \mathbb{P}\left( M_2 \cap M_0 \cap \left\{ \left[ \mathcal{KL}^{1/2} \left( Y_i, \theta_i \right) + \mathcal{KL}^{1/2} \left( Y_j, \theta_j \right) \right]^2 > \left[ \sqrt{z} / \varkappa - \varphi_0 \right]^2 \right\} \right) \\
	& = & \mathbb{P}\left( M_3 \cap M_0 \cap \left\{ \left[ \mathcal{KL}^{1/2} \left( \mathcal{Y}_i, \vartheta \right) - \mathcal{KL}^{1/2} \left(  \mathcal{Y}_j, \vartheta \right) \right]^2
	> \left[ \sqrt{z}/ \varkappa - \varphi_0 \right]^2 \right\} \right) \\
	&& + \, \mathbb{P}\left( M_4 \cap M_0 \cap \left\{ \left[ \mathcal{KL}^{1/2} \left( \mathcal{Y}_i, \vartheta \right) + \mathcal{KL}^{1/2} \left( \mathcal{Y}_j, \vartheta \right) \right]^2 > \left[ \sqrt{z} / \varkappa - \varphi_0 \right]^2 \right\} \right).
\end{eqnarray*}
Equation~\eqref{eq:KL-Taylor} leads on~$M_3 \cap M_0$ with appropriate parameters $\vartheta_1^{\ast}, \vartheta_2^{\ast} \in \Theta_{\varkappa}$ to
\begin{eqnarray*}
	&& \left[ \mathcal{KL}^{1/2} \left( \mathcal{Y}_i, \vartheta \right) - \mathcal{KL}^{1/2} \left(  \mathcal{Y}_j, \vartheta \right) \right]^2 \\
%	& \overset{\text{Eq.~\eqref{eq:KL-Taylor}}}{\leq} & \max \left\{ \frac{1}{2 I(\vartheta_1^{\ast})}, \frac{1}{2 I(\vartheta_2^{\ast})} \right\} \left[ \vert C(\mathcal{Y}_i) - C(\vartheta) \vert - \vert C(\mathcal{Y}_j) - C(\vartheta) \vert \right]^2 \\
	& \overset{M_3}{\leq} & \max \left\{ \frac{1}{2 I(\vartheta_1^{\ast})}, \frac{1}{2 I(\vartheta_2^{\ast})} \right\} \left[ \left( C(\mathcal{Y}_i) - C(\vartheta) \right) - \left( C(\mathcal{Y}_j) - C(\vartheta) \right) \right]^2 \\
	& \overset{\text{Eq.~\eqref{eq:varkappa}}}{\leq} & \varkappa^2 \mathcal{KL} \left( \mathcal{Y}_i, \mathcal{Y}_j \right).
\end{eqnarray*}
On~$M_4 \cap M_0$, we get in a uniform manner
\begin{eqnarray*}
	\left[ \mathcal{KL}^{1/2} \left( \mathcal{Y}_i, \vartheta \right) + \mathcal{KL}^{1/2} \left( \mathcal{Y}_j, \vartheta \right) \right]^2
%	& \overset{M_4}{\leq} & \max \left\{ \frac{1}{2 I(\vartheta_1^{\ast})}, \frac{1}{2 I(\vartheta_2^{\ast})} \right\} \left[ \left( C(\mathcal{Y}_i) - C(\vartheta) \right) + \left( C(\vartheta) - C(\mathcal{Y}_j) \right) \right]^2 \\
	\overset{\text{Eq.~\eqref{eq:varkappa}}}{\leq} \varkappa^2 \mathcal{KL} \left( \mathcal{Y}_i, \mathcal{Y}_j \right).
\end{eqnarray*}
Hence, we conclude that
\begin{eqnarray*}
	\mathbb{P}\left( \lbrace \mathcal{KL} (Y_i, Y_j) > z \rbrace \cap M_0 \right)
	& \leq & \mathbb{P}\left( \left\{ \varkappa^{2} \mathcal{KL} \left( \mathcal{Y}_i, \mathcal{Y}_j \right) > \left[ \sqrt{z} / \varkappa - \varphi_0 \right]^2 \right\} \cap M_0 \right) \\
	&=& \mathbb{P}\left( \left\{ \varkappa^2 \left[ \varkappa \, \mathcal{KL}^{1/2} \left( \mathcal{Y}_i, \mathcal{Y}_j \right) + \varphi_0 \right]^2 > z \right\} \cap M_0 \right).
\end{eqnarray*}
This terminates the proof.
\end{proof}

\begin{proof}[Justification of Claim~\ref{claim}]
Recall Notation~\ref{not:inhomVsHom}.
For simplicity, we concentrate on  the case $T = \mathrm{Id}$ as presumed in Proposition~\ref{prop:inhomPC}.
For $T \neq Id$, the assertion follows just as for $T = Id$, replacing in the proof of Proposition~\ref{prop:inhomPC} and in the following formulas the observations~$Y_i$ and~$\mathcal{Y}_i$ by the transformed observations~$T(Y_i)$ and~$T(\mathcal{Y}_i)$ for all $i \in \lbrace 1,...,n \rbrace$.
Recall from the proof of Lemma~\ref{lem:funcP2} that the assumed strict monotonicity of~$T$ ensures that the regions of monotonicity in Equation~\eqref{eq:KL-monotonicity} remain valid.

We know from Theorems~\ref{thm:PS21} and~\ref{thm:PS22} that the Kullback-Leibler divergence between the non-adaptive estimator and its expectation converges, in probability, at least with rate~$\overline{N}_i^{(k)}$ under homogeneity and under inhomogeneity. 
The choice of~$\lambda$ is in accordance with the homogeneous propagation condition by assumption. 
Hence, it compensates with probability $1 - \epsilon$ the impact of the adaptivity under homogeneity, and it only depends on the functions~$p^{(l)}_{\vartheta}$, $l=1,2,3$, see Lemma~\ref{lem:funcP2}.
Due to the assumed invariance of~$p^{(l)}_{\vartheta}$ with respect to~$\vartheta$, it holds $p^{(l)}_{\theta_i} = p^{(l)}_{\vartheta}$ for every $l=1,2,3$ and all $i \in \lbrace 1,...,n \rbrace$. 
Therefore, it suffices to increase the homogeneous bandwidth~$\lambda$ pursuant to the maximal impact of the local variability of~$\theta(.)$, but independent of the precise definition of~$\theta(.)$.

The latter effects the interplay of the observations and hence the adaptive weights, where we consider the random variables~$\lbrace s_{ij}^{(k)} \rbrace_{i,j}$, see Algorithm~\ref{algorithmMS} (page~\pageref{algorithmMS}). 
Proposition~\ref{prop:inhomPC} provides, on the set~$M_0$, an upper bound for the augmentation of the random variable $\mathcal{KL}\left( Y_i, Y_j \right)$ compared to $\mathcal{KL}(\mathcal{Y}_i, \mathcal{Y}_j)$.
It justifies the given choice of~$\lambda_{\varphi}$ for the iteration step $k=1$ if~$h^{(0)}$ satisfies $\overline{w}_{ij}^{(0)} = 0$ for all $X_i \neq X_j$. 
We seek for a generalization to other choices of~$h^{(0)}$ and the subsequent iteration steps. 

The adaptive estimator is defined as a weighted mean of the observations. 
Therefore, for all $k \in \lbrace 0,...,k^* \rbrace$, it holds
$M_0 \subseteq \lbrace \tilde{\theta}_i^{(k)}, \tilde{\vartheta}_i^{(k)} \in \Theta_{\varkappa}, i \in \lbrace 1,..., n \rbrace \rbrace \rbrace$,
where~$M_0$ is as in Proposition~\ref{prop:inhomPC}.
This enables on~$M_0$ the application of Equation~\eqref{eq:varkappa} (page~\pageref{eq:varkappa}) and Lemma~\ref{lem:PS52} with respect to the adaptive estimates.
We distinguish the same cases as in the proof of Proposition~\ref{prop:inhomPC}, recall Equations~\eqref{eq:M1-M2} and~\eqref{eq:M3-M4} and the corresponding upper bounds in Equations~\eqref{eq:KL-M2} and \eqref{eq:KL-M1}.
For the sake of brevity, we summarize both cases in one equation, using the operation~$\pm$.
Then, we get on the set~$M_0$ in a uniform manner as in the proof of Proposition~\ref{prop:inhomPC} that
\begin{eqnarray}
	s_{ij}^{(k)} & \leq & \varkappa^2 \tilde{N}_i^{(k-1)} \left[ \left| \mathcal{KL}^{1/2} \left( \tilde{\theta}_i^{(k-1)}, \mathcal{E} \tilde{\theta}_i^{(k-1)} \right) 
	\pm \mathcal{KL}^{1/2} \left( \tilde{\theta}_j^{(k-1)}, \mathcal{E} \tilde{\theta}_j^{(k-1)} \right) \right|
	 \right. \nonumber \\ 
	 && \left.
	+ \, \mathcal{KL}^{1/2} \left( \mathcal{E} \tilde{\theta}_i^{(k-1)}, \mathcal{E} \tilde{\theta}_j^{(k-1)} \right) \right]^2, \label {eq:sijZerlegung}
\end{eqnarray}
where~$\mathcal{E} \tilde{\theta}_i^{(k)}$ is as in Notation~\ref{not:mathcalE}.
The variability of the parameter function~$\theta(.)$ effects the third summand, which satisfies by Equation~\eqref{eq:varkappa} (page~\pageref{eq:varkappa}) and the convexity of the Kullback-Leibler divergence with respect to the first argument that
\[
	\max_{i,j} \mathcal{KL}\left( \mathcal{E} \tilde{\theta}_i^{(k-1)}, \mathcal{E} \tilde{\theta}_j^{(k-1)} \right)
	\leq \varkappa^2 \max_{i,j} \mathcal{KL} \left( \theta_i, \theta_j \right)
	\leq \varkappa^2 \, \varphi_0^2
\]
with $\varphi_0 = \varphi / \max_i \sqrt{\overline{N}_i^{(k^*)}}$.
The remaining term 
\begin{equation}\label{eq:sijZerl1}
	\sqrt{\tilde{N}_i^{(k-1)}} \left| \mathcal{KL}^{1/2} \left( \tilde{\theta}_i^{(k-1)}, \mathcal{E} \tilde{\theta}_i^{(k-1)} \right) 
	\pm \mathcal{KL}^{1/2} \left( \tilde{\theta}_j^{(k-1)}, \mathcal{E} \tilde{\theta}_j^{(k-1)} \right) \right| 
\end{equation}
forms the inhomogeneous analog of 
\begin{equation}\label{eq:sijZerl2}
	\sqrt{\tilde{N}_i^{(k-1)}} \left| \mathcal{KL}^{1/2} \left( \tilde{\vartheta}_i^{(k-1)}, \vartheta \right) 
	\pm \mathcal{KL}^{1/2} \left( \tilde{\vartheta}_j^{(k-1)}, \vartheta \right) \right|.
\end{equation}
However, the corresponding probability distributions cannot be compared as for the single observations since the probability distributions of~$\tilde{\vartheta}_l^{(k-1)}$ and~$\tilde{\theta}_l^{(k-1)}$, $l=i,j$, may differ considerably.
Nevertheless, it follows in the same lines as at the end of the proof of Proposition~\ref{prop:inhomPC} that
\begin{eqnarray*}
	\left| \mathcal{KL}^{1/2} \left( \tilde{\vartheta}_i^{(k-1)}, \vartheta \right) 
	\pm \mathcal{KL}^{1/2} \left( \tilde{\vartheta}_j^{(k-1)}, \vartheta \right) \right|
	\leq \varkappa \, \mathcal{KL}^{1/2} \left( \tilde{\vartheta}_i^{(k-1)}, \tilde{\vartheta}_j^{(k-1)} \right).
\end{eqnarray*}
Hence, Equation~\eqref{eq:sijZerl2} is controlled by~$\sqrt{\lambda}$, up to the factor~$\varkappa$.
Similarly, Equation~\eqref{eq:sijZerl1} mainly depends on the randomness of the observations. 
Admittedly, this cannot be proven due to the impact of the adaptive weights which are influenced by the variability of the inhomogeneous parameter function. 

Instead, we follow an inductive argumentation, considering the relation to the non-adaptive estimator.
The initialization of the algorithm by the non-adaptive estimator serves as the base clause.
Assuming that the adaptive weights in iteration step~$k$ are, with high probability, similar to the non-adaptive ones, we get that the divergence $\mathcal{KL} ( \tilde{\theta}_i^{(k)}, \mathcal{E} \tilde{\theta}_i^{(k)} )$ behaves similar to $\mathcal{KL} ( \overline{\theta}_i^{(k)}, \mathbb{E} \overline{\theta}_i^{(k)} )$. 
Additionally, we know from Theorems~\ref{thm:PS21} and~\ref{thm:PS22} that $\mathcal{KL} ( \overline{\vartheta}_i^{(k)}, \vartheta )$ and $\mathcal{KL} ( \overline{\theta}_i^{(k)}, \mathbb{E} \overline{\theta}_i^{(k)} )$ satisfy, in probability, the same rate of convergence.
The divergence $\mathcal{KL} ( \overline{\vartheta}_i^{(k)}, \vartheta )$ relates via the homogeneous propagation condition to the divergence $\mathcal{KL} ( \tilde{\vartheta}_i^{(k)}, \vartheta )$ and, as a consequence, to Equation~\eqref{eq:sijZerl2}, which we controlled by the constant~$\varkappa \, \sqrt{\lambda}$. 
This motivates together with Proposition~\ref{prop:inhomPC} and 
the invariance of the functions~$p^{(l)}_{\vartheta}$, $l=1,2,3$, with respect to $\theta \in \Theta$, 
the supposition that the impact of the variability of the parameter function on Equation~\eqref{eq:sijZerl1} is sufficiently small such that~$\varkappa \, \sqrt{\lambda}$ can still control it.
Then, we may conclude that the choice
\[
	\lambda_{\varphi} \geq \varkappa^4 \left[ \sqrt{\lambda} + \varphi \right]^2
\]
ensures in the next iteration step~$k+1$ the similarity of the adaptive and the non-adaptive weights,
yielding on~$M_0$ the desired behavior of~$\hat{\mathfrak{Z}}_{\lambda_{\varphi}}$.
The restriction to the set~$M_0$ leads to an increased probability level of $\epsilon + 2 p_{\varkappa}$ since
$\mathbb{P}(M_0^c) \leq 2 p_{\varkappa}$.
\end{proof}

\phantomsection \label{proof:PS22}
\begin{proof}[Proof of Theorem~\ref{thm:PS22}]
Let $q(u|v) = \mathcal{KL} (v, v+u)$ be as in Theorem~\ref{thm:PS61}.
The reparametrization $v = C(\theta)$ and $D(v) = B(\theta)$ yields
$D'(v) = \theta$, $D''(v) = 1/I(\theta)$, and
\begin{equation}\label{eq:KL(v)}
	\mathcal{KL} \left( \mathbb{P}_{v_1}, \mathbb{P}_{v_2} \right) 
	= D'(v_1) \left[ v_1 - v_2 \right] - \left[ D(v_1) - D(v_2) \right].
\end{equation}
This provides with the Taylor expansion that
\[
	q(u | v)
	= D( v + u ) - D(v) - u D'(v)
	= u^2 D''( v+ c u ) / 2,
\]
where the remainder is in Lagrange form, and $c \in [0,1]$ is chosen appropriately.
We set $\alpha := \varkappa^2 - 1$, and recall that $\overline{w}_{ij}^2 \leq \overline{w}_{ij}$ since $\overline{w}_{ij} \in [0,1]$.
For $c_1, c_2 \in [0,1]$ appropriate and all $i, j \in \lbrace 1,..., n \rbrace$, this yields condition~\eqref{eq:expBound}  in Theorem~\ref{thm:PS61} via
\begin{eqnarray*}
	q( \mu u \overline{w}_{ij} | v_j )
	& = & (\mu u \overline{w}_{ij})^2 D''(v_j + c_1 \mu u \overline{w}_{ij}) / 2 \\
	& \overset{\text{Eq.~\eqref{eq:varkappa}}}{\leq} & \mu^2 \overline{w}_{ij}^2 \varkappa^2 u^2 D''(\breve{v}_i + c_2 u) / 2 \quad
	\leq (1 + \alpha) \mu^2 \overline{w}_{ij} q( u | \breve{v}_i )
\end{eqnarray*}
if $C^{-1}(v_j + c_1 \overline{w}_{ij} \mu u), C^{-1}(\breve{v}_i + c_2 u) \in \Theta_{\varkappa}$.
The function~$C$ is strictly monotonic increasing, and the expectation satisfies $\breve{v}_i \in [ \min_j v_j, \max_j v_j ]$.
It holds by assumption that $C^{-1} (v_i) \in \Theta_{\varkappa}$ for all $i \in \lbrace 1,...,n \rbrace$ and, as a consequence, $C^{-1} (\breve{v}_i) \in \Theta_{\varkappa}$.
Therefore, it suffices to ensure that $C^{-1} (v_j + u) \in \Theta_{\varkappa}$ for all $j \in \lbrace 1,...,n \rbrace$ and $u \in \mathcal{U} (\overline{W}_i, z)$.
The assertion of Theorem~\ref{thm:PS61} remains valid if condition~\eqref{eq:expBound} is only satisfied for $u \in \mathcal{U} (\overline{W}_i, z)$ with $u := \overline{v}_i - \breve{v}_i$.
Hence, we restrict our analysis to the favorable realizations, where $C^{-1} (v_j + \overline{v}_i - \breve{v}_i) \in \Theta_{\varkappa}$ for all $i,j \in \lbrace 1,...,n \rbrace$ and some most favorable subset $\Theta_{\varkappa} \subseteq \Theta$.
The probability of the complementary set of realizations is bounded by the probability~$\breve{p}_{\varkappa}$ in Notation~\ref{not:pKappa2}, and we get by Theorem~\ref{thm:PS61} that
\[
	\mathbb{P} \left( \overline{N}_i \mathcal{KL}(\overline{\theta}_i, \mathbb{E} \overline{\theta}_i) > z \right) 
	\leq 2 e^{-z/\varkappa^2} + \breve{p}_{\varkappa},
\]
which leads to the assertion.
\end{proof}

\section*{Acknowledgements}

This work was partially supported by the 
Stiftung der Deutschen Wirtschaft (SDW). The author would like to thank Peter Math\'{e}, J\"{o}rg Polzehl, and Karsten Tabelow 
(WIAS Berlin) for helpful discussions.

\bibliographystyle{plainnat}
\bibliography{PS06-modelmiss}

\end{document}